\documentclass[conference]{IEEEtran}

\usepackage[T1]{fontenc}
\def\doi#1{\href{https://doi.org/\detokenize{#1}}{\url{https://doi.org/\detokenize{#1}}}}

\usepackage{hyperref}
\usepackage{amsmath}
\usepackage{amsfonts}
\usepackage{amssymb}
\usepackage{amsthm}
\usepackage{mathdots}
\usepackage{color}
\usepackage{proof}
\usepackage{bussproofs}
\usepackage{thmtools}
\usepackage{xspace}
\usepackage[shortlabels]{enumitem}
\usepackage{tikz}
\usepackage{graphicx}
\usepackage{adjustbox}

\newcommand{\Lan}{\mathcal{L}}

\newcommand{\For}[1][]{\mathrm{For}_{#1}}
\newcommand{\imp}{\supset}
\newcommand{\coimp}{\prec}
\newcommand{\Sa}{\Rightarrow} 
\newcommand{\Ra}{\Sa}

\newcommand{\BiInt}{\mathbf{BiInt}}

\newcommand{\SF}{\mathbf{S5}}
\newcommand{\LK}{\mathbf{LK}}
\newcommand{\LJ}{\mathbf{LJ}}
\newcommand{\GF}{\mathbf{G4}}

\newcommand{\RR}{\mathbf R}

\newcommand{\seq}{\Rightarrow}

\newcommand{\cont}[1]{\mathcal{C}(#1)}

\newcommand{\trans}[2]{
\vspace{3pt}
\noindent\fbox{
\begin{minipage}{.47\textwidth}
    #1\\
    \noindent\textbf{precondition:} #2
\end{minipage}}
\vspace{3pt}\\
}

\newcommand{\subst}[2]{\lbrack #1/#2\rbrack}

\newcommand{\context}[1]{\mathcal{C}(#1)}

\newcommand{\lleft}{\mathtt{l}}
\newcommand{\lright}{\mathtt{r}}
\newcommand{\lx}{\mathtt{x}}
\newcommand{\lbleft}[1]{\lleft#1}
\newcommand{\lbright}[1]{\lright#1}
\newcommand{\lbx}[1]{\lx#1}

\newcommand{\tms}[2]{\{#1\}^{#2}}
\newcommand{\mv}[1]{{#1}}

\newcommand{\SC}{\mathbb{S}}
\newcommand{\supp}[1]{\mathrm{supp}(#1)}

\newcommand{\princred}{\emph{principal reductions}}
\newcommand{\inversion}{\emph{reduction by inversion}}
\newcommand{\rightshift}{\emph{antecedent shift}}
\newcommand{\leftshift}{\emph{succedent shift}}
\newcommand{\leftanalytic}{\emph{analytic cutting on the left}}
\newcommand{\rightanalytic}{\emph{analytic cutting on the right}}
\newcommand{\renaming}{\emph{renaming of variables}}

\newcommand{\part}[4]{(#1:#2),(#3:#4)}
\newcommand{\Var}{\mathsf{Var}}
\newcommand{\subf}{\mathsf{subf}}

\newcommand{\ly}{\mathtt{y}}
\newcommand{\lby}[1]{\ly#1}

\newcommand{\red}[1]{\mathcal{E}_{#1}}

\newcommand{\hide}[1]{}

\newcommand{\mydots}[1]
{
\tikzset{every picture/.style={line width=0.75pt}} 
\begin{tikzpicture}[x=0.3pt,y=0.42pt,yscale=-1,xscale=1]

\draw  [draw opacity=0][fill={rgb, 255:red, 224; green, 224; blue, 224 }  ,fill opacity=1 ] (160,-10) -- (160,30) -- (71,58) -- (1,18) -- (-10,-10) -- cycle ;

\draw (82,25) node [anchor=base][inner sep=0.75pt]   [align=center] {\mbox{\footnotesize$#1$}};

\end{tikzpicture}
}

\usepackage{xcolor}
\newcommand{\rr}[1]{\\\textcolor{red}{#1}\\}

\newtheorem{definition}{Definition}
\newtheorem{lemma}[definition]{Lemma}
\newtheorem{corollary}[definition]{Corollary}
\newtheorem{example}[definition]{Example}
\newtheorem{theorem}[definition]{Theorem}
\newtheorem{remark}[definition]{Remark}
\newtheorem{fact}[definition]{Fact}

\newtheorem{claim}[definition]{Claim}

\begin{document}

\title{Cut-restriction: from cuts to analytic cuts}

\author{
\IEEEauthorblockN{Agata Ciabattoni}
\IEEEauthorblockA{\textit{TU Vienna} \\
Vienna, Austria \\
agata@logic.at}
\and
\IEEEauthorblockN{Timo Lang}
\IEEEauthorblockA{\textit{University College London} \\
London, United Kingdom \\
timo.lang@ucl.ac.uk}
\and
\IEEEauthorblockN{Revantha Ramanayake}
\IEEEauthorblockA{\textit{University of Groningen} \\
Groningen, Netherlands \\
d.r.s.ramanayake@rug.nl}
}


\maketitle

\begin{abstract}
Cut-elimination is the bedrock of  proof theory
with a multitude of applications from computational interpretations to proof analysis. It is also the starting point for 
important meta-theoretical investigations into decidability, complexity, disjunction property, interpolation, and more.
Unfortunately cut-elimination does not hold for the sequent calculi of most non-classical logics. 
It is well-known that the key to applications is the subformula property (a typical consequence of cut-elimination) rather than cut-elimination itself. 
With this in mind, we 
introduce cut-restriction, a
procedure to restrict arbitrary cuts to analytic cuts (when elimination is not possible). The algorithm applies to all sequent calculi satisfying language-independent and  simple-to-check conditions,
and it is obtained by adapting age-old cut-elimination.
%
Our work encompasses 
existing results in a uniform way, subsumes Gentzen's cut-elimination,
and 
establishes new analytic cut properties.
\end{abstract}


\section{Introduction}

The fundamental result in proof theory is cut-elimination. It is the algorithm that syntactically eliminates cuts from a sequent calculus proof and leads to a proof that has the \emph{subformula property}, i.e., it only contains formulas that are subformulas of the final statement.
Though potentially larger in size, cut-free proofs are much better behaved and more amenable to meta-theoretic investigation, as the space of proofs under consideration is greatly constrained. 
Gentzen’s motivation in the 1930’s was a ``finitistic" proof of consistency of arithmetic but the influence of cut-elimination goes far beyond that. 
From its interpretation as computation under the proofs-as-programs correspondence (e.g.~\cite{CH06}) 
to its role in proof analysis  (e.g.~\cite{BaazHLRS08}), 
it is by far the most fundamental proof transformation.
Cut-free calculi--the offspring of cut-elimination--are widely applied to prove metalogical properties of the underlying logics (e.g., decidability, upper bounds, various flavours of interpolation, and disjunction properties), and they are key to developing automated reasoning methods.

Given the influence of Gentzen's seminal result it was inevitable that it would be extended to more and more logics.
Cut-elimination was originally proved for the sequent calculi for classical and intuitionistic logic IL but interesting and useful logics continue to be introduced and so the program of developing cut-free calculi via cut-elimination was extended. The first significant obstacle was encountered in the early 1950's: how to eliminate cuts in the proof calculus for the modal logic $\SF$? 
In 1968, Mints~\cite{Min68} solved the problem but not using the sequent calculus: he proved cut-elimination for $\SF$ using a generalisation of the sequent calculus known today as the hypersequent calculus, e.g.~\cite{Avr96}. This ushered in a new era in proof theory: generalise the sequent calculus to obtain proof calculi with cut-elimination for various  logics of interest. 
Nowadays, numerous proof formalisms can be found that generalise the sequent calculus: nested, labelled, bunched, tree-hypersequent, display sequent calculi, and many more.

Let us return to the sequent calculus. 
%
While cut-elimination implies the subformula property, it is \textit{not} a necessary condition.
%
In particular, the subformula property still holds if one accepts \textit{analytic cuts}~\cite{Smu68}, i.e. a cut-rule where the cut-formula is a subformula of the conclusion. 
Indeed, Kowalski and Ono~\cite{KowOno17} show that the subformula property 
is \emph{equivalent} to the analytic cut property (every theorem has a proof whose cuts are analytic).
Notably, many decidability, complexity and interpolation arguments go through in the presence of analytic cuts (\cite{Tak92,OnoSano}). For instance,
Kowalski and Ono~\cite{KowOno17} show a host of results such as Craig interpolation, Halld\'{e}n completeness and Maximova variable separation, utilising the analytic cut property. 
As Smullyan~\cite{Smu68} describes it, ‘the real importance of cut-free proofs is not the elimination of cuts \textit{per se}, but rather that such proofs obey the subformula principle’.
%


As with generalised proof formalisms, the analytic cut property (or mild relaxations, as 
in~\cite{Tak01,Tak19,Tak20}) can serve as a substitute when cut-elimination in the sequent calculus fails. Moreover,
there are advantages in retaining the sequent calculus as a framework for meta-logical investigations over using a generalised formalism: it is the sequent calculus rather than any of these other formalisms that remains familiar to anyone with a passing knowledge of formal logic. Aside from its association with Gentzen's famous result, compelling reasons are its simplicity, the ready identifiability of a sequent with logical consequence (assumptions are on the left of the sequent and consequences on the right), the fact that the additional structure/meta-language in extended formalisms usually complicate meta-logical investigations, and the numerous results that have been proved using the sequent calculus.


What is missing is a general methodology for transforming sequent calculus proofs with arbitrary cuts into proofs with analytic cuts. 
In fact, 
almost all proofs of the subformula property in the literature have been semantic,  e.g.~\cite{Fit78,DAgoMon94,KowOno17,AvrLahav}.
%
%
However, semantic arguments lack an important feature (present in Gentzen's cut-elimination): it is not possible to extract an algorithm to eliminate non-analytic cuts from a proof in a stepwise manner. Indeed, the constructive nature of cut-elimination was important to satisfy Hilbert's requirement for a ``finitistic" proof of consistency. Even beyond this, it is crucial for ordinal analysis and computational interpretations.

The singular exception to the prevalent semantic proofs is Takano's 1992 paper~\cite{Tak92}, 
where he took up the challenge to obtain the analytic cut property by syntactic means for the sequent calculus of $\SF$ and some logics in its vicinity. He then proved analogous results for more modal logics (\cite{Tak01,Tak19,Tak20}) but this time via semantic method. While the result of~\cite{Tak92} is well-known, the syntactic method presented there seems to be virtually unknown.
This might be due to its intricate arguments, and because it is unclear how to visualise the transformations in~\cite{Tak92} and how they fit together (unlike Gentzen's transformations). Consequently, it is hard to determine whether the arguments transfer to other logics. 
Another syntactic solution for $\SF$ was recently obtained in Ciabattoni \textit{et al.}~\cite{CiaLanRam21} where proofs are first embedded into a hypersequent calculus, cut-elimination is applied there, and the hypersequent structure is then systematically removed in favour of analytic cuts. This gives an in-principle algorithm in the sequent calculus but the complex sequence of transformations in no way resembles Gentzen's reductions, and a high-level `picture' of the transformations is once again lacking.


The present paper provides a methodology to transform sequent calculus proofs with arbitrary cuts into proofs with analytic cuts.
We call the resulting method \emph{cut-restriction}, of which Gentzen's cut-elimination is a special, boundary case.
We focus on a class of propositional sequent calculi
well-behaved for our purpose and large enough to include a variety of 
interesting calculi.  
%
The algorithm of cut-restriction proceeds in a stepwise manner and   
is formulated using language-independent conditions along the lines of Belnap's~\cite{Bel82} 
conditions for cut-elimination. Checking these conditions is straightforward and does not require knowledge of the algorithm. 



Our work brings together scattered results in the literature, and provides a uniform way to prove
the analytic cut-property for a host of calculi (in particular, simplifying Takano's argument for $\SF$), including
the calculus {\bf L4}~\cite{Ono77}
for the strongest S5-type intuitionistic modal logic {\bf G4}, and its extension with a coimplication connective.
This work also
resolves the following open question for bi-intuitionistic logic (a conservative extension of IL):
\begin{quote}
The sequent calculus $\BiInt$~\cite{Rau74} fails cut-elimination but it is complete with analytic cuts. Semantic proofs of this result were presented independently by Kowalski and Ono~\cite{KowOno17}, and Avron and Lahav~\cite{AvrLahav}. 
Pinto and Uustalu~\cite{PintoU18} prove syntactically that certain infinitely many (co)implicational cut-formulas suffice for completeness but left open the problem of finding a constructive procedure for the analytic cut property in $\BiInt$.
\end{quote}


Summing up, the contribution of this paper is twofold.
\begin{itemize}
\item  We introduce the first proof transformation 
reducing arbitrary cuts to analytic cuts that applies to a large class of propositional sequent calculi. In doing so, we extend Gentzen's age-old transformations.
\item We provide easy-to-check sufficient conditions on the sequent calculus for analytic cut property.
\end{itemize}


%

\smallskip
\subsubsection*{Cut-restriction needs a novel idea}

At first sight it might seem reasonable to assume that cut restriction follows from some slight adaptation of cut-elimination.
We illustrate using the case of $\SF$ that this is \textit{not} the case.
The following presumes some knowledge of cut-elimination; the reader unfamiliar with this terminology is referred to Section~\ref{sec-idea}. First consider the cut below that is well-known~\cite{Ohn59} to be not eliminable in $\SF$: 
\begin{center}
\AxiomC{$\Box\lnot p\Ra \Box\lnot p$}
\RightLabel{(${\lnot}{\mathtt{r}}$)}
\UnaryInfC{$\Ra \lnot\Box\lnot p,\Box\lnot p$}
\RightLabel{$(5)$}
\UnaryInfC{$\Ra \Box\lnot\Box\lnot p,\Box\lnot p$}
\AxiomC{$p\Ra p$}
\RightLabel{(${\lnot}{\mathtt{l}}$)}
\UnaryInfC{$\lnot p,p\Ra$}
\RightLabel{$(T)$}
\UnaryInfC{$\Box\lnot p,p\Ra$}
\RightLabel{cut}
\BinaryInfC{$p\Ra \Box\lnot\Box\lnot p$}
\DisplayProof
\end{center}
The cut-formula~$\Box\lnot p$ is principal in the right premise of cut by the $(T)$ rule and it is non-principal (i.e. context) in the left premise.
Therefore the usual move in cut-elimination would be to lift the cut upward in the left premise which means a cut on
$\Ra \lnot\Box\lnot p,\Box\lnot p$ and $\Box\lnot p,p\Ra$ yielding $p \Ra \lnot\Box\lnot p$. However, we cannot apply $(5)$ now since that rule requires that every context formula is boxed.
%
Note that the cut in the proof diagram is analytic because $\Box\lnot p$ is a subformula of $p\Ra \Box\lnot\Box\lnot p$. We conclude that we cannot lift the cut upwards as in Gentzen's cut-elimination but if we are prepared to accept analytic cuts then 
nothing more needs to be done here. Of course, it still remains to show that analytic cuts suffice in \textit{all} situations. We want a constructive proof so a natural idea is to generalise cut-elimination by considering an arbitrary topmost non-analytic cut and seek transformations that make the cut-formula smaller until the cut disappears or
becomes analytic. However, this idea does not hold up in practice. Consider:
\begin{center}
\AxiomC{$\Box\Gamma\Ra \Box A, B $}
\RightLabel{$(5)$}
\UnaryInfC{$\Box\Gamma\Ra\Box A,\Box B$}
\AxiomC{$A,\Sigma\Ra\Pi$}
\RightLabel{$(T)$}
\UnaryInfC{$\Box A,\Sigma\Ra\Pi$}
\RightLabel{non-analytic cut}
\BinaryInfC{$\Box\Gamma,\Sigma\Ra\Box B,\Pi$}
\DisplayProof
\end{center}
The cut-formula is principal in the right premise by $(T)$ and non-principal in the left premise.
As in the previous example, the cut cannot be lifted up in either premise.
There we accepted the cut as it was analytic, but how to proceed if it is not? A new idea is needed. Even the briefest consideration of~\cite{Tak92} will provide an indication of the intricacy of Takano's solution.

\smallskip
The paper is organised as follows. The idea behind cut-restriction is discussed informally in Sec.~\ref{sec-idea} using $\BiInt$ as a case study.
Sec.~\ref{Sec:preliminaries} introduces the 
class of calculi we deal with ({\em standard sequent calculi}).
These are sequent calculi having all structural rules, and whose logical rules are analytic and introduce one connective at a time.
The main ingredients for the cut-restriction proof to go through are identified in Sec.~\ref{Sec:conditions}; they are formulated in terms of syntactic sufficient conditions to be satisfied by a standard sequent calculus. The general proof of cut restriction is contained in Sections~\ref{sec:cut-restriction} and~\ref{sec:mainproof}, and examples of calculi to which it applies are presented in Sec.~\ref{sec:applications}. Sec.~\ref{sec:classone}shows how Gentzen's cut-elimination is a special case of cut restriction.  

\section{A guided example}\label{sec-idea}

\newcommand{\topgamma}{top(\gamma)}
\newcommand{\topdeltaone}{top_1(\delta)}
\newcommand{\topdeltatwo}{top_2(\delta)}
\newcommand{\botgamma}{bot(\gamma)}
\newcommand{\botdelta}{bot(\delta)}


\noindent
{\em Gentzen's cut-elimination argument:} Stepwise reductions (`simplifying transformations') replace a cut with smaller cuts with respect to a well-founded relation. The cut-free proof follows from a transfinite induction.
The stepwise reductions come in two flavours: \textit{permutation} and
\textit{principal} reduction. The former shifts a cut one step upwards in either the left or the right premise.
Following repeated applications, the situation is reached of a cut in which the cut-formula is \textit{principal} (i.e.
created by the rule immediately above it) in both premises. The principal reduction is now used to replace that cut with cuts on proper subformulas. Principal reductions depend on the shape of the introduction rules and in some cases they can be hard to find. This is what happens with the modal rule in provability logic $GL$, for example: the change in polarity of the diagonal formula from conclusion to premise necessitates a highly intricate and customised principal reduction~\cite{Val83}.

Here we will consider cut-restriction for sequent calculi in which principal reductions are unproblematic. Therefore we shift our attention to permutation reductions. 

%
%
\emph{Permutation reductions fail}
if a rule 
cannot be permuted with the cut that follows it. The reason for the failure is that the rule 
cannot be applied \textit{after} the cut because the new premises 
conflict with the context restrictions of the rule.

Permutation reductions are unproblematic for 
Gentzen's calculus $\LK$ for classical logic as there are no context restrictions. 
In Maehara's calculus for intuitionistic logic~\cite{Takeuti:87}---a multiple-conclusion sequent calculus
obtained by replacing the right implication rule in $\LK$ with that of the intuitionistic calculus $\LJ$ (cf. $(\to \lright)$ in Fig.~\ref{fig:examplerules})
---some permutation reductions 
\textit{do} fail as $(\to \lright)$  does not permit any context on the right (the principal formula must appear alone).
E.g., try to permute the following cut upwards in the right premise:
\begin{center}
\AxiomC{\mydots{\gamma}}
\noLine
\UnaryInfC{$\Gamma \seq \Delta, A \to B$}
\AxiomC{\mydots{\delta}}
\noLine
\UnaryInfC{$C, \Sigma, A \to B \seq D$}
\RightLabel{$(\to \lright)$}
\UnaryInfC{$\Sigma, A \to B \seq C\imp D$}
\RightLabel{$(cut)$}
\BinaryInfC{$\Gamma,\Sigma \seq \Delta, C\to D$}
\DisplayProof
\end{center}
This means the following transformation:
\begin{center}
\AxiomC{\mydots{\gamma}}
\noLine
\UnaryInfC{$\Gamma \seq \Delta, A \to B$}
\AxiomC{\mydots{\delta}}
\noLine
\UnaryInfC{$C, \Sigma, A \to B \seq D$}
\RightLabel{$(cut)$}
\BinaryInfC{$\Gamma,C, \Sigma \seq \Delta, D$}
\DisplayProof
\end{center}
We are stuck as we cannot apply $(\to\lright)$ to this sequent when $\Delta$ is non-empty as the rule does not permit any right context.

The solution here is known: repeatedly shift this cut upwards in the left premise $\gamma$ until the cut-formula is weakened or introduced by $(\to \lright)$ (we refer to it as \textit{critical inference}).
$$\infer[(\to \lright)]{\Gamma'\Rightarrow A\to B}{A,
 \Gamma'\Rightarrow B} $$
 In the latter case, it is crucial that $A\to B$ is the sole formula on the right.
Only at this point we do lift the cut upward in the right premise. It does not cause any issue since the critical inference does not introduce any context on the right. We ultimately obtain a cut 
whose cut-formula is principal in both premises. Now use a principal reduction to replace the cut by smaller cuts (on $A$ and $B$) and cut-elimination follows.

To set the scene for later, view this as first tracing the cut-formula till principal in both premises and applying the principal reduction. Now, proceed down the trace in the right premise, and then the left. A sufficient condition for the latter is that for every rule instance, if it has a context formula $A\to B$ on the right then its substitution by arbitrary multisets on the left \textit{and} right is also a rule instance (we say that $\to$ is \textit{rightable} Def.~\ref{def:substitution}); for the former it is the ability to substitute a context formula $A\to B$ on the left with the context of $(\to\lright)$ (\textit{weakly leftable}) i.e. we are anticipating substitution with the critical inference context.
A special case of Theorem~\ref{thm-classone} states that every standard sequent calculus (Definition~\ref{def:standardSC}) whose every connective is rightable and weakly leftable has cut-elimination.


\subsubsection*{Cut-restriction}
Consider now the sequent calculus $\BiInt$~\cite{Rau74} for bi-intuitionistic logic.
Bi-intuitionistic logic is a conservative extension of intuitionistic logic that introduces a connective $\coimp$ that is residuated with $\lor$ in the same sense that $\to$ and $\land$ are residuated.
%
%
Formally, $\BiInt$ is obtained from the Maehara calculus for IL by 
adding the rules $(\coimp \lleft)$ and $(\coimp \lright)$ in Fig.~\ref{fig:examplerules}. %
Crucially, the $(\coimp\lleft)$ rule permits a context on the right but not on the left. Consequently, the permutation reduction upward in the left premise $\gamma$ that we applied before is not possible in $\BiInt$. 
From the perspective of this paper, the $\to$ connective is not rightable in $\BiInt$ and the sufficient conditions for cut-elimination in Theorem~\ref{thm-classone} are not met. This is not surprising as some theorems of $\BiInt$ have no cut-free proof (see~\cite{KowOno17}).
It is time to move from cut-elimination to cut-restriction up to analytic cuts. Trace the ancestors of the cut-formula all the way to their critical inferences (similar ``tracing back" arguments are used, e.g.,
in cut-elimination proofs for the sequent calculus~\cite{Restall,BaazL06} or
display calculus~\cite{Bel82}).
Critical inferences split the proof into a top and a bottom part.
A simplified situation where there is a single critical inference in each of the left and right premise is shown below. 
\begin{prooftree}
\AxiomC{\mydots{\topgamma}}
\noLine
\UnaryInfC{$\Gamma',A\Sa B$}
\RightLabel{$(\to\!\!\lright)$}
\UnaryInfC{$\Gamma'\seq A\to B$}
\noLine
\UnaryInfC{\mydots{\botgamma}}
\noLine
\UnaryInfC{$\Gamma \seq A \to B, \Delta$}
\AxiomC{\mydots{\topdeltaone}}
\noLine
\UnaryInfC{$\Sigma' \seq A, \Pi'$}
\AxiomC{\hskip -20pt\mydots{\topdeltatwo}}
\noLine
\UnaryInfC{\hskip -30pt$\Sigma', B \seq \Pi'$}
\RightLabel{$(\to\!\!\lleft)$}
\BinaryInfC{$\Sigma', A \to B \seq \Pi'$}
\noLine
\UnaryInfC{\mydots{\botdelta}}
\noLine
\UnaryInfC{$\Sigma, A \to B \seq \Pi$}
\RightLabel{$(cut)$}
\BinaryInfC{$\Gamma,\Sigma \seq \Delta, \Pi$}
\end{prooftree}
%
\begin{figure*}
\label{fig-guided}
\begin{prooftree}
\AxiomC{$\Gamma'\seq \Gamma'$}
\noLine
\UnaryInfC{\mydots{\botgamma \subst{\Gamma'}{A\to B}}}
\noLine
\UnaryInfC{$\Gamma \seq \Gamma', \Delta$}

\AxiomC{\mydots{\topgamma}}
\noLine
\UnaryInfC{$\Gamma',A\Sa B$}
\AxiomC{\mydots{\topdeltaone}}
\noLine
\UnaryInfC{$\Sigma' \seq A, \Pi'$}
\RightLabel{$(cut)$}
\BinaryInfC{$\Gamma',\Sigma\Sa B,\Pi'$}
\AxiomC{\mydots{\topdeltatwo}}
\noLine
\UnaryInfC{$\Sigma', B \seq \Pi'$}
\RightLabel{$(cut)$}
\BinaryInfC{$\Sigma',\Gamma'\Sa\Pi'$}
\noLine
\UnaryInfC{\mydots{\botdelta \subst{\Gamma'}{A\to B}}}
\noLine
\UnaryInfC{$\Sigma,\Gamma'\seq \Pi$}
\RightLabel{$(cut)$}
\BinaryInfC{$\Gamma,\Sigma \seq \Delta, \Pi$}
\end{prooftree}
\caption{Illustration of the transformation for cut-restriction up to analytic cuts in $\BiInt$}
\end{figure*}
To simplify further, assume that $\Gamma'$ consists of a single formula. Now transform this as shown in Fig.~\ref{fig-guided}. In a nutshell:
\begin{itemize}
    \item The original cut on $A\to B$ is replaced by a cut on $\Gamma'$, the context of the critical inference in $\gamma$.   
    \item We replace all ancestors of $A\to B$ in $bot(\delta)$ by $\Gamma'$, leading to a leaf of the form $\Sigma',\Gamma'\Sa\Pi'$. This leaf can be proved via the usual principal case reductions on $top(\gamma)$, $top_1(\delta)$ and $top_2(\delta)$.
     \item We also replace ancestors of $A\to B$ in $bot(\gamma)$ by $\Gamma'$. This leads to a ``trivial'' leaf of the form $\Gamma'\Sa\Gamma'$.
\end{itemize}
Let us explain why the substitutions are well-defined.
\begin{description}
\item[$\botdelta \subst{\Gamma'}{A\to B}$:] 
As $bot(\delta)$ contains $A\to B$ as a context formula on the left, the rule $(\coimp\lleft)$ could not have been applied on that branch. In every other rule instance, we can replace such a $A\to B$ by the context $\Gamma'$ of $(\to\lright)$. This is a property of $\to$ called \emph{weakly leftable}.


\item[$\botgamma \subst{\Gamma'}{A\to B}$:]
As $bot(\gamma)$ contains $A\to B$ as a context formula on the right, the $(\to\lright)$ was not applied on that branch. In other rule instances, we can replace the $A\to B$ by $\Gamma'$. Note $\Gamma'$ originally occurred on the left (so it `swaps sides' here). This property of $\to$ is \emph{inverse rightable}.

\end{description}


In general, the newly introduced cut on $\Gamma'$ is not analytic. However, by suitably preprocessing the proof and selecting the uppermost non-analytic one, we can show that the introduced cut is either analytic, or $\Gamma'$ is a \emph{proper} subformula of $C$. In the latter case, we have improved the situation; we repeatedly transform the proof to eventually obtain an analytic cut.

Moving from the simplified situation above to the general case, two complications arise. 
First, the presence of the contraction rule means that we might have to trace more than one occurrence of $A\to B$ so the transformation sketched above has to be modified accordingly. This is done in a rather standard way by using the \emph{multicut rule} instead of the cut rule.
Second, $\gamma$ may contain multiple critical inferences, all with different contexts $\Gamma'$ containing any number (including zero) of formulas. We will introduce cuts on \emph{all} these formulas. 
%

%
In summary, we retain the principal reductions, replace permutation reductions by tracing the predecessors of the cut-formulas along branches ensuring that they remain well-defined when substituted by new formulas (from contexts of critical inferences), and apply cut to remove these formulas.




\section{Standard Calculi}
\label{Sec:preliminaries}
We start by formalising the class of calculi we consider.

Fix a language $\Lan$ consisting of logical connectives, each with some integer arity $\geq 0$.  A connective of arity $0$ is called a \emph{constant}. The set $\For$ of formulas is generated in the usual way from variables ($x,y,\ldots$) and the connectives in $\Lan$. 
The \emph{principal connective} in a formula is its outermost connective. A formula with principal connective $\circ$ will be denoted $A^\circ$.

To simplify the notation in the cut-reduction proof, we will work with \emph{labelled formulas} (\emph{$\ell$-formulas} for short) of the form $\lbleft{A}$ and $\lbright{A}$. Intuitively, $\lbleft{A}$ denotes an occurrence of the formula $A$ on the left (antecedent) of a sequent and $\lbright{A}$ denotes an occurrence on the right (succedent) (cf.\ ``signed formulas''~\cite{AvrLahav}). The notion of a \emph{(proper) subformula} is lifted to $\ell$-formulas by ignoring the labels. 
%
We identify Gentzen \emph{sequents }$A_1\ldots,A_n\seq B_1,\ldots, B_m$ with multisets of $\ell$-formulas ($\ell$-multisets for short) $\{\lbleft{A_1},\ldots,
\lbleft{A_n},\lbright{B_1},$ $\ldots,$ $\lbright{B_m}\}$. 

Uppercase greek letters ($\Gamma,\Delta,\Sigma,\ldots)$ denote $\ell$-multisets, and uppercase latin letters ($A,B,C\ldots$) are used both for formulas and $\ell$-formulas. 
$\ell$-multisets containing formulas all labelled $\lleft$ (resp. $\lright$) are denoted $\lbleft{\Gamma}$ (resp. $\lbright{\Gamma}$). $\Gamma^p$, $p \geq 0$ is the $p$-fold union of the multiset $\Gamma$ with itself (e.g.\ $\{A,A,B\}^2=\{A,A,A,A,B,B\}$), while $\Gamma,\Delta$ is the multiset union of $\Gamma$ and $\Delta$.
 By convention, $\Gamma^0=\emptyset$. 
The support of $\Gamma$, denoted by $\supp{\Gamma}$, is the set of elements that occur at least once in $\Gamma$.
%
%
%
%
Also, $\Gamma^\ast$ is a  \emph{contraction} of $\Gamma$ if $\supp{\Gamma}=\supp{\Gamma^\ast}$ and every element appears as least as often in $\Gamma$ as it does in $\Gamma^\ast$. 


The class of sequent rules under consideration is formalised below in an abstract manner.
These are logical rules having a single \emph{principal formula} in their conclusion, and whose premises contain proper subformulas of this formula 
(\emph{auxiliary formulas}).
%
The rules have an \emph{additive} context, i.e.\ the same additional formulas appear both in the premises and in the conclusion of each rule. Crucially, this context can be restricted, meaning that only certain formulas are allowed. 





\begin{definition}[simple  rules]\label{standard-rule}
Let $\circ\in\Lan$. A \emph{simple left rule $\RR$ for $\circ$} is a pair $(\Lambda(\RR),\context{\RR})$ such that:
\begin{itemize}
    \item $\Lambda(\RR)$ is a set of tuples of the form $(\lbleft{C^\circ}, \Lambda_1,\ldots,\Lambda_M)$ for some fixed arity $M\geq 0$. $\lbleft{C^\circ}$ is the \emph{principal formula}, and each $\Lambda_{m}$ ($1\leq m\leq M$) is an $\ell$-multiset of proper subformulas of $C^\circ$ called \emph{auxiliary formulas}.
    \item $\context{\RR}$ is a set of $\ell$-formulas called \emph{context restriction}
\end{itemize}

An \emph{instance} of $\RR$ then is a figure
\begin{prooftree}
    \AxiomC{$\Gamma,\Lambda_1$}
    \AxiomC{$\ldots$}
    \AxiomC{$\Gamma,\Lambda_M$}
    \TrinaryInfC{$\Gamma,\lbleft{C^\circ}$}
\end{prooftree}
where $(\lbleft{C^\circ}, \Lambda_1,\ldots,\Lambda_m)\in\Lambda(\RR)$ and $\Gamma$ is a multiset of $\ell$-formulas from $\context{\RR}$; we call $\Gamma$ the \emph{context} of the instance.

We require that $\Lambda(\RR)$ and $\context{\RR}$ are closed under uniform substitution, and $\Lambda(\RR)$ is total in the first component.\footnotemark{}
%
\emph{Simple right rules} are defined analogously: replace $\lbleft{C^{\circ}}$ with $\lbright{C^{\circ}}$.
\end{definition}


%
\footnotetext{\emph{Uniform substitution}: if we replace all occurrences of a variable $x$ by a formula $A$ in an instance of $\Lambda(\RR)$ or $\context{\RR}$, we get  an instance of $\Lambda(\RR)$ resp.\ $\context{\RR}$. 
\emph{Total in the first component}: for every formula $C^\circ$, there is a tuple in $\Lambda(\RR)$ whose first component is $\lbleft{C^\circ}$.}
%


%
A context restriction $\context{\RR}$ prescribes the type of formulas that can be used as a context in rule instances. 
We say that \emph{$\RR$ has no context restriction} if 
$\context{\RR}$ is maximal, that is if 
$\context{\RR}$ is the set of all $\ell$-formulas.

Rules are usually presented by a schematic figure rather than a formal specification. Fig.~\ref{fig:examplerules} contains many such schemata (in the standard two-sided presentation). The example below illustrates how these fit into the framework of Def.~\ref{standard-rule}.   

\begin{example}
Consider the rules in Fig.~\ref{fig:examplerules}.
Neither $(\land\lleft)$ nor $(\land\lright)$ has a context restriction. We have
\begin{align*}
    \Lambda(\land\lleft)=\{(\lbleft{(A\land B)},\{\lbleft{A},\lbleft{B}\})\mid A,B\in\For\}&\text{ and } \\
     \Lambda(\land\lright)=\{(\lbright{(A\land B)},\{\lbright{A}\},\{\lbright{B}\})\mid A,B\in\For\}&.
\end{align*}
Context restrictions of the rules $(\lright\to)_M$, $(5)$ and $(\to\lright)_\Box$ are 
\begin{align*}
    \{\lbleft{F}\mid F\in\For\}&,
    \\
    \{\lbleft{\Box F}\mid F\in\For\}\cup \{\lbright{\Box F}\mid F\in\For\}&\text{, and}
    \\
    \{
    \lbleft{F}\mid F\in\For\}\cup \{\lbright{\Box F}\mid F\in\For
    \}&\text{ respectively}.
\end{align*}
The usual rules for the constant $\bot$ (once again, no context restriction) are simple rules with
\begin{align*}
    \Lambda(\bot\lleft)=\{(\lbleft{\bot})\}\text{ and } 
     \Lambda(\bot\lright)=\{(\lbright{\bot},\emptyset)\}.
\end{align*}
\end{example}

\begin{definition}[standard sequent calculus]
\label{def:standardSC}
A \emph{standard $\Lan$-calculus} $\SC$ consists of the \emph{initial sequents} 
\begin{center}
$\lbleft{x},\lbright{x}\quad (id)$
\end{center}
where $x$ is any variable, 
together with:
\begin{itemize}
    \item the \emph{structural rules} of weakening~$w$, contraction~$c$ and \emph{multicut} $mcut$    
    \[
\infer[(w)]{\Gamma,\Delta}{\Gamma}
\qquad
\infer[(c)]{\Gamma^\ast}{\Gamma}
\qquad
\infer[(mcut)]{\Gamma,\Delta}
    {
    \Gamma,\{\lbright{C}\}^p
    &
    \Delta,\{\lbleft{C}\}^q
    }
\]
where $\Gamma^\ast$ is a contraction of $\Gamma$ and $p,q\neq 0$.
    \item a 
    simple left rule $(\circ\lleft)$
    and a simple right rule $(\circ\lright)$ for every $\circ\in\Lan$
\end{itemize}
\end{definition}

\noindent The formula $C$ in $(mcut)$ is called the \emph{cut formula}.



\begin{example}[some standard calculi]\label{ex:calculi} Let $\Lan_0=\{\land,\lor,\to,\bot\}$. 
\begin{itemize}
    \item Gentzen's calculus $\LK$ for classical propositional logic (consisting
    of the simple rules in the first two rows of Fig.~\ref{fig:examplerules}) is a standard $\Lan_0$-calculus.
    \item Maehara's multiple-conclusion calculus for intuitionistic logic~\cite{Takeuti:87} is obtained by replacing in $\LK$ the rule $(\to\lright)$ with $(\to\lright)_M$ (Fig.~\ref{fig:examplerules}) of intuitionistic logic.
    \item The $\mathcal{L}_0\cup\{\coimp\}$-calculus $\BiInt$~\cite{Rau74} for bi-intuitionistic logic extends Maehara's calculus with the rules $(\coimp \lleft)$ and $(\coimp \lright)$ in Fig.~\ref{fig:examplerules}.
    \item The $\mathcal{L}_0\cup\{\Box\}$-calculus $\SF$ is obtained by adding to the $\LK$ calculus the rules $(T)$ and $(5)$.
\end{itemize}
\end{example}
It is useful to distinguish between proofs and deductions.

\begin{definition}[deductions and proofs in standard calculi] 
A \emph{deduction} of $\Gamma$ from $\Omega$ in a standard calculus is a tree of sequents rooted in $\Gamma$ (the \emph{endsequent}) that is composed of rule instances, and every leaf is either an initial sequent or contained in $\Omega$. 
A deduction from $\Omega=\emptyset$ is called a \emph{proof}. $\Gamma$ is \emph{provable} if there is a proof with endsequent $\Gamma$.
\end{definition}

\begin{definition}[analytic cut]
An instance of $(mcut)$ is \emph{analytic} if the cut formula is a subformula of some formula in the conclusion of the instance.
\end{definition}

\begin{definition}
    A deduction is \emph{cut-free} if it does not use the rule $(cut)$. A deduction is \emph{locally analytic} if all instances of cut in it are analytic.
$\SC$ \emph{admits cut-elimination} if every provable sequent has a cut-free proof. $\SC$ has the \emph{analytic cut property} if every provable sequent has a locally analytic proof. 
\end{definition}

It is immediate that locally analytic deductions in a standard calculus have the 
\emph{subformula property} (every formula occurring in the deduction is a subformula of the endsequent). 

\begin{figure*}
\fbox{
\begin{minipage}{\textwidth}
\noindent
\[
\begin{array}{cccc}
\infer[(\bot\lleft)]{\Gamma,\bot\Sa\Delta}{}
&
\infer[(\bot\lright)]{\Gamma\Sa\bot,\Delta}{\Gamma\Sa\Delta}
&
\infer[(\land\lleft)]{\Gamma,{A\land B}\Sa\Delta}{\Gamma,{A},{B}\Sa\Delta}
&
\infer[(\land\lright)]{\Gamma\Sa{A\land B},\Delta}
    {
    \Gamma\Sa{A},\Delta
    &
    \Gamma\Sa{B},\Delta
    }
\\\\
\infer[(\lor\lleft)]{\Gamma,A\lor B\Sa\Delta}
    {
    \Gamma,A\Sa\Delta
    &
    \Gamma,A\Sa\Delta
    }
&
\infer[(\lor\lright)]{\Gamma\Sa A\lor B,\Delta}
    {\Gamma\Sa A,B,\Delta}
&
\infer[(\to \lleft)]{\Gamma, A \to B\Sa\Delta}
    {
    \Gamma\Sa A,\Delta
    &
    {\Gamma}, B\Sa\Delta
    }
&
  \infer[(\to \lright)]{{\Gamma}\Sa {A\to B},\Delta}
    {
    \Gamma, {A} \Sa {B},\Delta
    } 
\\\\


 \infer[(\to \lright)_M]{{\Gamma}\Sa {A\to B}}
    {
    \Gamma, A\Sa B
    } 
    &
     \infer[(\coimp \lleft)]{{A \coimp B}\Sa\Delta }{A\Sa {B},\Delta}
&
  \infer[(\coimp \lright)]{ {\Gamma}\Sa A\coimp B,\Delta}{\Gamma\Sa A,\Delta
     & {\Gamma}, B\Sa\Delta
     }
&
\infer[(\land_i,\lleft)_{i=1,2}]{\Gamma,A_1\land A_2\Sa\Delta}
    {
    \Gamma,A_i\Sa\Delta
    }
      
     \\\\
 \infer[(T)]{{\Gamma}, {\Box A}\Sa\Delta}{{\Gamma}, A,\Sa\Delta}
 &
 \infer[(5)]{{\Box \Gamma}\Sa {\Box A},\Box\Delta }{{\Box \Gamma}\Sa {A},\Box\Delta}
&
\infer[(\to \lright)_\Box]{{\Gamma}\Sa {A\to B}, \Box \Delta}
        {{\Gamma}, {A}\Sa
     {B}, \Box \Delta}
&
\infer[(\coimp \lleft)_\Box]{{A \coimp B}, {\Box \Gamma} \Sa\Delta }{A, {\Box \Gamma} \Sa {B},\Delta}
\end{array}
\]
\end{minipage}
}
\caption{Some simple rules 
}
\label{fig:examplerules}
\end{figure*}

\section{Sufficient Conditions}
\label{Sec:conditions}
Fix a standard calculus $\SC$. In order to generalise the case study in Section~\ref{sec-idea}, we introduce here abstract conditions for $\SC$ to satisfy cut-restriction, and hence the analytic cut-property. 

The first two conditions 
are very familiar to proof theorists.
Axiom expansion is also known as the \emph{identity theorem}~\cite{Pfe10}. The principal case reduction corresponds to Belnap's condition (C8) for cut-elimination in display calculus~\cite{Bel82}. 

\begin{definition}[axiom expansion]
\label{def:expansion}
$\SC$ satisfies \emph{axiom expansion} if $
\lbleft{A},\lbright{A}$ has a cut-free proof for every formula $A$.
\end{definition} 
\begin{definition}[principal case reductions]
\label{def:princred}
A standard calculus satisfies \emph{principal case reductions}
if whenever $(\lbleft{C^\circ}, \Lambda_1,\ldots,\Lambda_M)\in\Lambda(\circ\lleft)$ and $(\lbright{C^\circ}, \Sigma_1,\ldots,\Sigma_N)\in\Lambda(\circ\lright)$ for some connective $\circ$, there is a deduction of the empty sequent from the sequents $\Lambda_1,\ldots,\Lambda_M,\Sigma_1,\ldots,\Sigma_N$ 
(each is a multiset of auxiliary formulas)
using only structural rules. 
%

\end{definition}

\begin{example}
Let us check the principal case reductions for $\land$ (cf. Fig.~\ref{fig:examplerules}). We have $(\lbleft{(A\land B)},\{\lbleft{A},\lbleft{B}\})\in\Lambda(\land\lleft)$ and $(\lbright{(A\land B)},\{\lbright{A}\},\{\lbright{B}\})
\in\Lambda(\land\lright)$ and 
\[
    \AxiomC{$\lbright{B}$}
    \AxiomC{$\lbright{A}$}
    \AxiomC{$\lbleft{A},\lbleft{B}$}
    \RightLabel{$(mcut)$}
    \BinaryInfC{$\lbleft{B}$}
    \RightLabel{$(mcut)$}
    \BinaryInfC{$\emptyset$}
    \DisplayProof
\]
\end{example}

\begin{definition}[consistency]
\label{def:consistency}
    $\SC$ is \emph{consistent} if it does not prove the empty sequent.
\end{definition}

This property is needed in the proof of the main theorem (only) to replace atomic cuts with analytic atomic cuts (case $(A2)$). 
While consistency is sometimes targeted as a corollary of cut-elimination (dating to Gentzen's pursuit of a ``finitistic" consistency proof for arithmetic), there is a much simpler and direct way to obtain it: exhibit a model that is closed under the axioms and rules of the calculus, and falsifies at least one formula of the logic. Many modal logics can be shown consistent e.g. by observing that their axioms and rules hold in a Kripke model consisting of a single reflexive world.

It is well-known that all calculi in Ex.~\ref{ex:calculi} are consistent
, satisfy axiom expansion and admit principal case reductions. 

\begin{definition}[invertibility]
\label{def:invertibility}
    $\circ$ is \emph{left-invertible} if the following holds:     
    If $(\lbleft{C}^\circ,\Lambda_1,\ldots,\Lambda_M)\in\Lambda(\circ\lleft)$ then for every $\ell$-multiset $\Gamma$ and every proof $\beta$ of $\Gamma,\lbleft{C^\circ}$ there is a proof $\beta_m$ of $\Gamma,\Lambda_m$ ($m\leq M$) satisfying the following:
    \begin{itemize}
        \item If $\beta$ is cut-free, then so is $\beta_m$
        \item If $\beta$ is locally analytic, then $\beta_m$ is locally analytic apart from possibly some cuts on proper subformulas of $C$
    \end{itemize} \emph{Right-invertible} is defined analogously: replace $\lbleft$ with $\lbright$.
\end{definition}


The condition in Def.~\ref{def:invertibility} will be used in the main proof. 
The following sufficient condition for invertibility 
is simpler to check in practice. 
See the appendix for a proof.
 \begin{lemma}
 \label{lem:invertible}
 A connective $\circ$ satisfying the conditions below is left-invertible (conditions for right-invertible are analogous).
 \begin{enumerate}
     \item For every $C^\circ$ there is a unique $(\Lambda_1,\ldots,\Lambda_M)$ such that $(\lbleft{C^\circ},\Lambda_1,\ldots,\Lambda_M)\in\Lambda(\circ\lleft)$
     \item Whenever $\lbleft{C^\circ}\in\context{\RR}$ for some simple rule $\RR$ and $(\lbleft{C^\circ},\Lambda_1,\ldots,\Lambda_M)\in\Lambda(\circ\lleft)$  then also $\supp{\Lambda_m}\subseteq\context{\RR}$ for every $m\leq M$.
\end{enumerate}
 \end{lemma}

Note that the uniqueness assumption in Lemma~\ref{lem:invertible} is satisfied for $\land$ with the rule $(\land\lleft)$ but not if we use their non-invertible variant $(\land_i,\lleft)_{i=1,2}$ (cf. Fig.~\ref{fig:examplerules}).
In our framework they would amount to a rule $\RR$ with $\Lambda(\RR)=\{(\lbleft{(A_1\land A_2)},\{\lbleft{A_i}\})\mid A_1,A_2\in\For,i=1,2\}$.

If all connectives are left- and right-invertible and $\SC$ satisfies principal case reductions, then cuts on arbitrary formulas can be reduced to cuts on variables. However, this situation rarely occurs.
The crux of this paper are the weaker \textit{substitution properties}, motivated and defined below,
%
%
which guarantee that non-principal occurrences of cut formulas in a proof can be replaced by certain other formulas. First observe that the permutation of a cut above a rule corresponds to a substitution:

\begin{center}\AxiomC{$\Gamma,\lbright{C^\circ}$}
\AxiomC{$\lbleft{C^\circ},\Lambda_1$}
\RightLabel{$\RR$}
\UnaryInfC{$\lbleft{C^\circ},\Lambda$}
\BinaryInfC{$\Gamma,\Lambda$}
\DisplayProof\quad
\AxiomC{$\Gamma,\lbright{C^\circ}$}
\AxiomC{$\lbleft{C^\circ},\Lambda_1$}
\BinaryInfC{$\Gamma,\Lambda_1$}
\RightLabel{?$\RR$?}
\UnaryInfC{$\Gamma,\Lambda$}
\DisplayProof
\end{center}
In the figure on the right the cut on $\lbleft{C}^\circ$ has been lifted above the instance of $\RR$, yielding a new instance where $\lbleft{C}^\circ$ has been substituted with $\Gamma$. In general, such substitutions can fail to be legal if $\RR$ has a context restriction: we might have $\lbleft{C}^\circ\in\context{\RR}$ but $B\notin \context{\RR}$ for some $B$ in $\Gamma$. The (weakly) leftable/rightable properties in Def.~\ref{def:substitution} assert that `nothing goes wrong' when such substitutions are carried out.
These properties are reformulations of known sufficient conditions for cut-elimination.
In more detail, the leftable property will ensure that any cut on $\lbleft{C}^\circ$ can be lifted above $\RR$. The weakly leftable property ensures that such a lifting is possible when $\lbright{C}^\circ$ is principal in the left premise of cut, and consequently $\Gamma$ contains only formulas from $\context{\circ\lbright}$. To make this substitution legal, we therefore require that $\context{\circ\lbright}\subseteq\context{\RR}$.

However, in some cases the context restrictions are such that a cut cannot be lifted to the point that it is principal in both premises (irrespective of the strategy of lifting). This is the crucial case that is a main concern of this paper. To handle this situation, we introduce a novel proof transformation that replaces the cut with analytic cuts. This transformation relies on the legality of new substitution properties that we call \emph{inverse leftable} and \emph{inverse rightable}. 
The latter property is that if $\lbright{C^\circ}\in\context{\RR}$, then any formula in $\context{\circ\lbright}$ `swapped' to the other side is  in the context of $\RR$. 
In other words, the substitution of the formula $\lbright{C^\circ}$ with a formula from $\context{\circ\lbright}$ whose side (label) is swapped leads to a new instance of $\RR$.
To state this property precisely we define inversion of labels, $\ell$-formulas and sets thereof as follows: 
$\bar{\lleft}=\lright$, $\bar{\lright}=\lleft$, $\overline{\lbx{A}}=\bar{\lbx{}} A$ and $\overline{\mathcal{C}}=\{\overline{\lbx{A}} | \lbx{A}\in\mathcal{C}\}$.
 The inverse rightable property is used to construct the leftmost sub-derivation in Fig~\ref{fig-guided} (observe the substitution $[G'/A \to B]$) of an analytic cut. In particular, by swapping the sides of formulas that we had in the original proof and ultimately cutting on them, we are able to replace the arbitrary cut with analytic cuts. 

\begin{definition}[substitution properties] Let $\circ$ be a connective of a standard calculus $\SC$.
It has the stated property when for every formula $C^\circ$ and every simple rule~$\RR$:
\label{def:substitution}
\begin{enumerate}
    \item (\emph{leftable})  $\lbleft{C^\circ}\in\context{\RR}$ implies that $\RR$ has no context restriction.
    (\emph{rightable})  $\lbright{C^\circ}\in\context{\RR}$ implies that $\RR$ has no context restriction.
    
    \item (\emph{weakly leftable}) $\lbleft{C^\circ}\in\context{\RR}$ implies $\context{\circ\lright}\subseteq\context{\RR}$. 
    
    (\emph{weakly rightable}) $\lbright{C^\circ}\in\context{\RR}$ implies $\context{\circ\lleft}\subseteq\context{\RR}$.
    
    \item (\emph{inverse leftable}) $\lbleft{C^\circ}\in\context{\RR}$ implies $\overline{\context{\circ\lleft}}\subseteq\context{\RR}$.
    
    (\emph{inverse rightable}) $\lbright{C^\circ}\in\context{\RR}$ implies $\overline{\context{\circ\lright}}\subseteq\context{\RR}$.
     
\end{enumerate}
\end{definition}

\begin{fact}
If a connective is leftable (rightable), then it is also weakly and inverse leftable (rightable).  
\end{fact}


We show our conditions at work in various standard calculi.

\begin{example}[Maehara's calculus]
   \label{Maehara}
    In Maehara's calculus, all connectives $\circ$ are rightable: If $\lbright{C^\circ}\in\context{\RR}$ then $\RR\neq(\rightarrow \lright)$, as $(\rightarrow \lright)$ does not permit context formulas labelled $\lright$, and all other rules $\RR$ have no context restriction.
    However $\to$ is {\em neither} leftable, as we cannot replace the context 
    formula $\lbleft{(A\rightarrow B})$ in an instance of $(\to \lright)$ with any $\lbright{\Sigma}$,
    {\em nor} inverse leftable, for the same reason. $\rightarrow$ is instead weakly leftable, as $\context{\to\lright}$ contains only formulas labelled $
    \lleft$ and these are allowed in the context of any other rule.

    \end{example}

    \begin{example}[$\BiInt$]
    \label{BI-Int}
As in Maehara's calculus, $\to$ is weakly leftable 
as the additional connective $\coimp$ does not pose problems: $(\coimp \lleft)$ has no left context so the condition $\lbleft{(A\rightarrow B}) \in\context{\coimp \lleft}$ implies $\context{\to \lright}\subseteq\context{\coimp \lleft}$ trivially holds, and
$(\coimp \lright)$ has no context restriction. However, $\to$ is not rightable unlike in Maehara's calculus as    
    $\lbright{(A\rightarrow B})$ might appear in the context of a rule instance $(\coimp \lleft)$ but would not admit a substitution with any $\lbleft{\Gamma}$. 
    Nevertheless, $\to$ is inverse rightable, as
    formulas in $\context{\to \lright}$ are left formulas, and the replacement of a
    $\lbright{(A\rightarrow B})$ in the context with 
    right formulas (note the switch
    from left to right) works for every rule. In particular, for $(\to \lright)$ the condition $\lbright{(A \to B)}\in\context{\to \lright}$ implies $\overline{\context{\to \lright}}\subseteq\context{\to \lright}$ holds trivially since $\context{\to \lright}$ contains no right formulas. 
    %
    $\coimp$ behaves symmetrically: $\coimp$ is weakly rightable and 
    inverse rightable.
    \end{example}

    \begin{example}[$\SF$ modalities]
    \label{ex:modalities}
    We discuss below diverse calculi containing the $(T)$ and $(5)$ rules in Fig~\ref{fig:examplerules}.
    Later, each will be seen to have cut-restriction under our general conditions.
\begin{enumerate}
    \item 
    In $\SF$ all connectives $\circ\neq\Box$ are both leftable and rightable. 
    In contrast, $\Box$ is neither leftable nor rightable 
    since $C^\Box$ \emph{can} appear in the left or right context of the $(5)$ rule and this rule has context restrictions (that only boxed formulas are allowed). $\Box$ is 
    not weakly rightable because
    $\lbright{C^\Box}\in\context{5}$ implies $\context{T}\subseteq\context{5}$ does not hold since $\context{T}$ contains non-boxed formulas but $\context{5}$ does not. As boxed formulas
    are permitted in the left and right context of every rule, it follows that $\Box$ is weakly leftable 
    and inverse rightable (note that $\overline{\context{5}}=\context{5}$, so the inversion does not matter here).

    \item Introduced in~\cite{Ono77}, the $\Lan_0\cup\{\Box\}$-calculus $\GF$ is obtained by replacing in Maheara's calculus
    (Ex.~\ref{Maehara}) the rule $(\to \lright)$ with $(\to \lright)_\Box$ from Fig.~\ref{fig:examplerules} 
and adding the $\SF$ rules $(T)$ and $(5)$. $\GF$ is sound and complete for L4, the strongest S5-type intuitionistic modal logic. 
Cut-elimination is known to fail for $\GF$, see~\cite{OnoSano}.
As in Maehara's calculus,  
all connectives $\circ\neq\Box$ are rightable, and moreover $\to$ is weakly leftable.
As in $\SF$, boxed formulas are permitted in the left and right context of any rule (here it becomes important that we use $(\to \lright)_\Box$) and so 
$\Box$  is weakly leftable and inverse rightable.

\item Let us consider the extension of {\bf G4} with ``well behaving" rules for $\coimp$, i.e.
$(\coimp \lright)$ and $(\coimp \lleft)_\Box$ from  Fig~\ref{fig:examplerules}. 
We call the resulting calculus $\BiInt^{\SF}$.
As in {\bf G4}, $\to$ is weakly leftable in  $\BiInt^{\SF}$ (the rules for $\coimp$ do not pose problems as $\lbleft{(A \to B)} \not\in \mathcal{C}(\coimp \lleft)_\Box$ and the rule $(\coimp \lright)$ has no context restriction) and inverse rightable (due to the presence of the context with $\Box$ formulas on the left in $(\coimp \lleft)_\Box$), 
while $\coimp$ is weakly rightable and inverse left-subtitutable. $\Box$ is weakly leftable and inverse rightable.

    \end{enumerate}
\end{example}

The following example exhibit a calculus whose $\Box$ modality does not satisfy any of our substitution properties.
\begin{example}
Introduced in~\cite{Ono77}, the calculus {\bf G3} is obtained by adding to Maheara's calculus the $\SF$ rules $(T)$ and $(5)$ from Fig.~\ref{fig:examplerules}. 
{\bf G3} does not admit cut-elimination.
In {\bf G3}, $\Box$ is neither weakly leftable (as a context formula $\lbleft{C^\Box}$ in $(\to\lright)$ cannot be replaced by any context of the $(5)$ rule, which might also contain boxed formulas labelled right), nor inverse leftable (for the same reason).
  \end{example}

\section{Cut-restriction}
\label{sec:cut-restriction}

We state our main theorem that presents cut-restriction yielding analytic cuts.
This is Theorem~\ref{thm:mainclasstwo}. The crucial Reduction Lemma is proved in the next section. Adapting Takano's terminology, we consider two classes of connectives that allow
for reducing arbitrary cuts to smaller cuts (class 1), and to analytic cuts (class 2).
\begin{definition}
A connective $\circ$ in a standard calculus is \emph{class~1} if it is one of the following:
\begin{enumerate}
    \item left-invertible and right-invertible
    \item leftable and weakly rightable
    \item rightable and weakly leftable
\end{enumerate}
It is \emph{class~2} if it is class 1 or one of the following:

\begin{enumerate}\setcounter{enumi}{3} 
    \item weakly leftable and inverse rightable
    \item weakly rightable and inverse leftable
\end{enumerate}
\end{definition}



\begin{definition}
[class 2 calculus]\label{def:classtwo}
    A standard calculus is \emph{class~2} if it is consistent, satisfies principal case reductions and axiom expansion, and every connective in it is class~2.
\end{definition}

Verifying membership is not modular with respect to language extensions for class~1 nor class~2 (see Remark~\ref{rem:class12}).


\begin{theorem}[Main Theorem]
\label{thm:mainclasstwo}
    Every class~2 standard calculus has the analytic cut property.
\end{theorem}

The proof 
uses Gentzen-style proof transformations that replace non-analytic cuts with either analytic or ``smaller'' cuts.

\begin{definition}[inductive measures]\mbox{}
The \emph{degree} of a multicut in a proof is the number of symbols in its cut formula, and its \emph{rank} is the number of sequents above it.

A proof is \emph{$(d,r)$-reduced} if every non-analytic multicut in it has  degree~$\leq d$, and those non-analytic multicuts of maximal degree $d$ have rank~$<r$ and are not below any other non-analytic multicut.
%
\end{definition}

Below is a schematic representation of a $(d,r)$-reduced proof with the restrictions on its non-analytic multicuts.

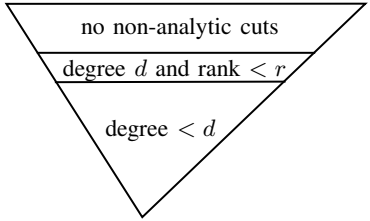
\begin{center}
\tikzset{every picture/.style={line width=0.75pt}} 

\begin{tikzpicture}[x=0.75pt,y=0.75pt,yscale=-.57,xscale=.65]

\draw   (239.84,209.51) -- (134.42,20.63) -- (412.95,20.48) -- cycle ;
\draw    (158,64) -- (374,64) ;
\draw    (172,90) -- (351,89.5) ;

\draw (176,67) node [anchor=north west][inner sep=0.75pt]   [align=left] {\small degree $\displaystyle d\ $and rank $\displaystyle < r$};
\draw (209,120) node [anchor=north west][inner sep=0.75pt]   [align=left] {\small degree $\displaystyle < d$};
\draw (190,30) node [anchor=north west][inner sep=0.75pt]   [align=left] {\small no non-analytic cuts};

\end{tikzpicture}

\end{center}


\noindent The following serves as main lemma in the proof of Thm.~\ref{thm:mainclasstwo}.

\begin{lemma}[Reduction Lemma]
\label{thm:classtwo}
Let $\beta$ be a proof in a class~2 calculus which is locally analytic apart from a single non-analytic multicut of degree $d$ and rank $r$ as its last inference. Then there is a $(d,r)$-reduced proof $\beta'$ of the same endsequent.
\end{lemma}

We devote the next section to a proof of the Reduction Lemma. First observe that Lemma~\ref{thm:classtwo} implies Thm.~\ref{thm:mainclasstwo}.%
\begin{proof}[Proof of Thm.~\ref{thm:mainclasstwo}]

It suffices to show that an arbitrary subproof $\beta$ ending in an uppermost non-analytic cut can be replaced by a locally analytic proof of the same endsequent. 
The proof proceeds by induction on the pair (degree,rank) of the uppermost non-analytic cut under the usual  lexicographic ordering.
%
Suppose that this non-analytic cut has (degree,rank) $(d,r)$ and suppose that the claim holds for every smaller pair. 
Apply the Reduction Lemma to replace $\beta$ with a $(d,r)$-reduced proof $\beta'$ (every uppermost non-analytic cut in $\beta'$ has degree $\leq d-1$, or degree~$d$ and rank~$<r$). Moreover $\beta'$ cannot contain a non-analytic cut of degree $d$ below another non-analytic cut of degree $d$ since the lower cut would violate the condition ``not below any other non-analytic multicut'' in the definition of $(d,r)$-reduced. Thus every non-analytic cut of degree $d$ in $\beta'$ is uppermost and can be eliminated by the IH to obtain $\beta''$ where each non-analytic cut has degree $\leq d-1$. Repeatedly applying the IH to uppermost non-analytic cuts---the degree of the lower cuts are unchanged after each application---we ultimately obtain a locally analytic proof. 
\end{proof}


%

\section{Proof of the Reduction Lemma}
\label{sec:mainproof}

Picture the lowermost multicut in $\beta$ like this:
\begin{align}
\label{eq:lowermost}
\infer[(mcut)]{\Gamma,\Delta}
    {
    \deduce{\Gamma,%
    \tms{\lbright{C}}{p}
    }{\mydots{\gamma}}
    &
    \deduce{\Delta,\tms{\lbleft{C}}{q}}{\mydots{\delta}}
    }
\end{align}
Recall that the superscripts $p$ and $q$ indicate multiplicities. By assumption, $\gamma$ and $\delta$ are locally analytic. Let $\RR(\gamma)$ and $\RR(\delta)$ denote the last rule in $\gamma$ and $\delta$ respectively.

We present a collection of  reductions that transform $\beta$ into a $(d,r)$-reduced proof $\beta'$. The first group of reductions, called \emph{Gentzen-style reductions} (A1), are well-known ingredients of many 
cut-elimination proofs. Some extra care has to be taken even in these familiar reductions as, unlike in cut-elimination, we cannot assume that $\gamma$ and $\delta$ are cut-free.

The second group (A2) consists of reductions that are peculiar to \emph{cut-restriction}: They do not necessarily decrease the degree or rank of the multicut, but might replace it with new analytic multicuts.

The required reductions depend on $\RR(\gamma)$, $\RR(\delta)$ and the properties of the outermost connective $\circ$ of $C$ (if $C$ is not a variable). In a final step (B), we argue that some reduction is always applicable. This establishes the Reduction Lemma.

It will be important in these reductions to be able to ``trace back'' occurrences of the multicut formula $C$. This is achieved by defining an \emph{immediate ancestor relation} between formula occurrences in the premise of each rule $\RR$ and formula occurrences in its conclusion as follows  (cf.~``congruence" in~\cite{Bel82}):
\begin{itemize}
    \item If $\RR=(w)$ (resp.\ $\RR=(c)$), then as the immediate ancestor relation we can take any injective (resp.\ surjective) function that maps $\ell$-formulas in the premise to the same $\ell$-formula in the conclusion.
    \item If $\RR=(mcut)$, the immediate ancestor relation is the obvious map from $\ell$-formulas in $\Gamma$ and $\Delta$ in the premises to $\ell$-formulas in $\Gamma,\Delta$ in the conclusion. The cut formulas are not immediate ancestors of any $\ell$-formula.
    \item If $\RR$ is any simple rule, the immediate ancestor relation relates context $\ell$-formulas in the premise(s) to identical $\ell$-formulas in the conclusion context, and auxiliary $\ell$-formulas in the premise to the principal formula.
\end{itemize}
The \emph{ancestor relation} is then the reflexive transitive closure of the immediate ancestor relation. We will call an \emph{ancestor of multicut} any ancestor of the cut formula $C$ in $\beta$'s lowermost multicut. In other words, an ancestor of multicut is an ancestor of any $C$ occuring in $\{\lbright{C}\}^p$ of $\{\lbleft{C}\}^q$ in~(\ref{eq:lowermost}).

The ancestor relation features in the following
lemma.

\begin{lemma}[Substitution Lemma]
\label{lem:subst}
Let
\[
    \infer{\Delta,\{C\}^q}
    	{
    	\Delta_1,\tms{C}{q_1}
    	&
    	\ldots
    	&
    	\Delta_n,\tms{C}{q_n}
    	}
\]
be an instance of a rule $\RR\neq(id)$ in a standard calculus where $C$ is an $\ell$-formula, $\{{C}\}^{q_i}$ marks the immediate ancestors of $\{{C}\}^q$ in the $i$th premise, and if $\RR$ is a simple rule then its principal formula does not appear in $\{C\}^q$.

For any $\ell$-multiset $\Gamma$ (with $\supp{\Gamma}\subseteq\context{\RR}$, in case of simple rules)
%
the following is also an instance of $\RR$:
\[
    \infer{\Delta,\Gamma^q}
    	{
    	\Delta_1,\Gamma^{q_1}
    	&
    	\ldots
    	&
    	\Delta_n,\Gamma^{q_n}
    	}
\]
\end{lemma}
\begin{proof}
If $\RR$ is a simple rule then $\{{C}\}^q$ is part of the context by hypothesis so  $q_1=\ldots=q_n=q$. If $\RR$ is weakening (contraction), then $n=1$ and $q_1\leq q$ ($q_1\geq q$). If $\RR$ is multicut, then $n=2$ and $q_1+q_2=q$. In all cases, the statement follows directly from inspection of the respective rule.
\end{proof}

We describe the reductions and their preconditions for their applicability.
\textit{The starting point is the proof diagram in (\ref{eq:lowermost})}.

\subsection*{(A1) Gentzen-style reductions}




\trans{\princred}
{
$C=C^\circ$, an ancestor of multicut is principal both in $\RR(\gamma)$ and in $\RR(\delta)$.
}
So $\beta$ concludes as follows:
\newsavebox\pleft
\sbox\pleft{
$
\AxiomC{\mydots{top_m(\gamma)}}
\noLine
\UnaryInfC{$\Gamma,\{\lbright{C}\}^{p-1},\Lambda_m$}
\DisplayProof
$
}
\newsavebox\pright
\sbox\pright{%
$%
\AxiomC{\mydots{top_n(\delta)}}
\noLine
\UnaryInfC{$\Delta,\{\lbleft{C}\}^{q-1},\Sigma_n$}
\DisplayProof%
$
}
\begin{prooftree}
    \AxiomC{$\left[\usebox\pleft\right]_m$}
    \RightLabel{$(\circ\lright)$}
    \UnaryInfC{$\Gamma,\{\lbright{C}\}^{p-1},\lbright{C}$}
    \AxiomC{$\left[\usebox\pright\right]_n$}
    \RightLabel{$(\circ\lleft)$}
    \UnaryInfC{$\Delta,\{\lbleft{C}\}^{q-1},\lbleft{C}$}
    \RightLabel{$(mcut)$}
    \BinaryInfC{$\Gamma,\Delta$}
\end{prooftree}
Henceforth we denote by  $[\ldots]_s$ a family of deductions indexed by the variable $s$. In the case above, we have $m\leq M$, $n\leq N$,  $(\lbright{C^\circ},\Lambda_1,\ldots,\Lambda_M)\in\Lambda(\circ\lright)$ and $(\lbleft{C^\circ},\Sigma_1,\ldots,\Sigma_N)\in\Lambda(\circ\lleft)$. We split the construction of $\beta'$ into two steps.

(1) We construct a proof $top_m^*(\gamma)$ of $\Gamma,\Delta,\Lambda_m$ ($m\leq M$). If $p-1\neq 0$, $top_m^*(\gamma)$ is obtained as:

\begin{prooftree}               \AxiomC{\mydots{top_m(\gamma)}}
    \noLine
    \UnaryInfC{$\Gamma,\{\lbright{C}\}^{p-1},\Lambda_m$}
    \AxiomC{\mydots{\delta}}
    \noLine
    \UnaryInfC{$\Delta,\{\lbleft{C}\}^{q-1},\lbleft{C}$}
    \RightLabel{$(mcut)$}
    \BinaryInfC{$\Gamma,\Delta,\Lambda_m$}
\end{prooftree}
Note that while the degree of this multicut remains $d$, its rank is strictly smaller than $r$. Thus $top_m^*(\gamma)$ is $(d,r)$-reduced. If $p-1=0$, then $top_m^*(\gamma)$ is defined to be $top_m(\delta)$ followed by weakening with $\Delta$. By a symmetric construction we obtain a $(d,r)$-reduced proof $top_n^*(\delta)$ of $\Gamma,\Delta,\Sigma_n$.

(2) Now combine these proofs into $\beta'$ as follows (we are using principal case reductions Def~\ref{def:princred}):
 
\newsavebox\topk
\savebox\topk{
$
\AxiomC{\mydots{top_m^*(\gamma)}}
\noLine
\UnaryInfC{$\Gamma,\Delta,\Lambda_m$}
\DisplayProof
$
}

\newsavebox\topl
\savebox\topl{
$
\AxiomC{\mydots{top_n^*(\delta)}}
\noLine
\UnaryInfC{$\Gamma,\Delta,\Sigma_n$}
\DisplayProof
$
}
\begin{prooftree}
    \AxiomC{$\left[\usebox\topk\right]_m$}
    \AxiomC{$\left[\usebox\topl\right]_n$}
    \dashedLine
    \RightLabel{$(mcut)$'s }
    \BinaryInfC{$\Gamma,\Delta$}
\end{prooftree}

By definition of a simple rule, the cuts coming from the principal case reductions are on formulas in $\supp{\Lambda_m}$ and $\supp{\Sigma_n}$, which in turn are proper subformulas of $C$. Thus the cuts are of degree $<d$. It follows that $\beta'$ is $(d,r)$-reduced.

\trans{\inversion}{$C=C^\circ$, $\circ$ left and right invertible}
Pick
$(\lbright{C^\circ},\Lambda_1,\ldots,\Lambda_M)\in\Lambda(\circ\lright)$ and $(\lbleft{C^\circ},\Sigma_1,\ldots,\Sigma_N)\in\Lambda(\circ\lleft)$
(by the definition of simple rule, there are always such tuples), 
and extend $\gamma$ and $\delta$ with contractions and weakenings to locally analytic proofs of $\Gamma,\Delta,\lbright{C}^\circ$ and $\Gamma,\Delta,\lbleft{C}^\circ$. 


By invertibility we obtain proofs $top^*_m(\gamma)$ of $\Gamma,\Delta,\Lambda_m$ and $top^*_n(\delta)$ of $\Gamma,\Delta,\Sigma_n$ for every $m\leq M$ and $n\leq N$ whose only multicuts are on proper subformulas of $C$. In particular, they are $(d,r)$-reduced. Now construct the $(d,r)$-reduced $\beta'$ as in the second step of the \princred{} case.

\trans{\rightshift}{
No ancestor of multicut is principal in $\RR(\delta)$ and one of the following holds:
\begin{enumerate}[(a)]
\item\label{assleftable} $ C=C^\circ$, $\circ$ is leftable 
\item\label{assweaklyleftable}  $ C=C^\circ$, $\circ$ is weakly leftable,  and $\supp{\Gamma}\subseteq\cont{\circ_\lright}$
\end{enumerate}}
We first isolate a single \emph{special case}, namely that $\RR(\delta)$ is again a multicut on $C$. In this case the end of $\beta$ runs as follows ($s,t>0$, $q=q_1+q_2$ and $\Delta=\Delta_1,\Delta_2$):

\[
\infer[(mcut)]{\Gamma,\Delta_1,\Delta_2}
    {
    \deduce{\Gamma,\tms{\lbright{C}}{p}}{\mydots{\gamma}}
    &
    \infer[(mcut)]{\Delta_1,\Delta_2,\{\lbleft{C}\}^{q_1+q_2}}
    	{
    	\deduce{\Delta_1,\{\lbleft{C}\}^{s+q_1}}{\mydots{\delta_1}}
    	&
    	\deduce{\Delta_2,\{\lbright{C}\}^{t},\tms{\lbleft{C}}{q_2}}{\mydots{\delta_2}}
    	}
    }
\]
$\beta'$ is then obtained as follows: Combine $\gamma$ and $\delta_1$ via a multicut to a proof of $\Gamma,\Delta_1$ 
followed by weakening. This new multicut has degree $d$ but rank $<r$, and so $\beta'$ is $(d,r)$-reduced.

So let us now assume the \emph{standard case} that $\RR=\RR(\delta)$ is not a multicut on $C$.
Note also that $\RR\neq(id)$ as $C$ is not a variable. The end of $\beta$ runs as follows:
\begin{align}
    \label{eq:beforeshift}
\infer[(mcut)]{\Gamma,\Delta}
    {
    \deduce{\Gamma,\tms{\lbright{C}}{p}}{\mydots{\gamma}}
    &
    \infer[\RR]{\Delta,\{\lbleft{C}\}^q}
    	{
    	\deduce{\Delta_1,\tms{\lbleft{C}}{q_1}}{\mydots{\delta_1}}
    	&
    	\ldots
    	&
    	\deduce{\Delta_n,\tms{\lbleft{C}}{q_n}}{\mydots{\delta_n}}
    	}
    }
\end{align}
where in the premises of $\RR$ we have highlighted all ancestors of multicut. $\beta'$ then is the proof
\begin{align}
\label{eq:aftershift}
\infer[(c)]{\Gamma,\Delta}
{
\infer[\RR]{{\Gamma}^{q},
\Delta}
	{
	\deduce{\Delta_1,{\Gamma}^{q_1}}{\mydots{\delta_1'}}
    	&
    	\ldots
    	&
    	\deduce{\Delta_n,{\Gamma}^{q_n}}{\mydots{\delta_n'}}
	}
}
\end{align}
 and $\delta_i'$ is defined as follows. 
If $q_i\neq 0$ then $\delta_i'$ is
\[
\infer[(w)]{{\Gamma}^{q_i},\Delta_i}
{
\infer[(mcut)]{\Gamma,\Delta_i}
	{
	\deduce{\Gamma,\tms{\lbright{C}}{p}}{\mydots{\gamma}}
	&
	\deduce{\Delta_i,\tms{\lbleft{C}}{q_i}}{\mydots{\delta_i}}
	}
 }
\]
This multicut is of grade $d$ but rank $<r$. If $q_i=0$ then $\delta_i':=\delta_i$. In any case $\delta_i'$ is $(d,r)$-reduced.

Let us argue that in~(\ref{eq:aftershift}) we indeed have an instance of $\RR$. By the Substitution Lemma (Lemma~\ref{lem:subst}) this follows directly if $\RR$ is not a simple rule, so assume it is. We have to show that $\supp{\Gamma}\subseteq \context{\RR}$. Indeed, by inspecting $\beta$ we see that $\lbleft{C^\circ}\in\context{\RR}$. Thus by assumption \ref{assleftable}, we get $\supp{\Gamma}\subseteq \context{\RR}$. If we assume \ref{assweaklyleftable}, we can observe  $\supp{\Gamma}\subseteq\context{\circ\lright}\subseteq \context{\RR}$.

It remains to establish that $\beta'$ is $(d,r)$-reduced. Knowing that each $\delta_i'$ is $(d,r)$-reduced, the only critical case is in fact that $\RR$ in~(\ref{eq:aftershift}) is a non-analytic multicut on some formula $D$. By assumption $\RR$ in~(\ref{eq:beforeshift}) is analytic. Clearly $\RR$ can lose its analyticity via the subsitution of $\Gamma$ for $\lbleft{C}$ only if $D$ was a subformula of $C$. As we have dealt with \emph{special case} $C=D$ separately, $D$ must in fact be a proper subformula of $C$. But then the degree of $\RR$ is $<d$, and so $\beta'$ is $(d,r)$-reduced.

\trans{\leftshift}{No ancestor of multicut is principal in $\RR(\gamma)$ and one of the following holds:
\begin{enumerate}[(a)]
\item  $ C=C^\circ$, $\circ$ is rightable 
\item  $ C=C^\circ$, $\circ$ is weakly rightable,  and $\supp{\Gamma}\subseteq\cont{\circ_\lleft}$
\end{enumerate}
}
Symmetric to \rightshift.


\vspace{-0.127cm}
\subsection*{(A2) Reductions introducing analytic cuts}

\trans{\renaming}{$C$ is a variable.}
Choose some formula $A$ appearing as a subformula in $\Gamma$ or $\Delta$ (at least one such formula exists because of the consistency assumption). We obtain $\beta'$ as follows: First, we replace all occurrences of the variable $x=C$ in $\beta$ by $A$. Note that this makes the lowermost multicut analytic. As all rules of a standard calculus apart from the initial sequents are closed under uniform substitution, we have obtained a deduction. Now some initial sequent $\lbleft{x},\lbright{x}$ in $\beta$ may have become a non-initial leaf $\lbleft{A},\lbright{A}$. In this case we equip such leaves with a cut-free proof using the axiom expansion property.

Note that pre-existing analytic multicuts in $\beta$ do not become non-analytic through uniform substitution. We conclude that $\beta'$ is locally analytic, and therefore in particular $(d,r)$-reduced.

\trans{\leftanalytic}{$C=C^\circ$, $\circ$ is weakly leftable and inverse rightable}
\label{analyticcutting}%
%
\newsavebox\innerleftproof
\savebox\innerleftproof{
$
\AxiomC{\mydots{ top_{\mv{im}}(\gamma)}}
\noLine
\UnaryInfC{$\Gamma_{\mv{i}},\tms{\lbright{C}}{p_{\mv{i}}},\Lambda_{\mv{im}}$}
\DisplayProof
$
}
\newsavebox\outerleftproof
\savebox\outerleftproof{
$\AxiomC{$\left[\usebox\innerleftproof\right]_m$}
\RightLabel{$(\circ\lright)_i$}
\UnaryInfC{$\Gamma_\mv{i},\tms{\lbright{C}}{p_\mv{i}},\lbright{C}$}
\DisplayProof$
}
\newsavebox\innerrightproof
\savebox\innerrightproof{
$\AxiomC{\mydots{top_{jn}(\delta)}}
\noLine
\UnaryInfC{$\Delta_{\mv{j}},\tms{\lbleft{C}}{q_{\mv{j}}},\Sigma_{\mv{jn}}$}
\DisplayProof$
}
\newsavebox\outerrightproof
\savebox\outerrightproof{
$\AxiomC{$\left[\usebox\innerrightproof\right]_n$}
\RightLabel{$(\circ\lleft)_j$}
\UnaryInfC{$\Delta_\mv{j},\tms{\lbleft{C}}{q_\mv{j}},\lbleft{C}$}
\DisplayProof$
}
\hskip -14pt This transformation requires a global argument and, together with the symmetric \rightanalytic{}, is the most involved one. It generalizes the argument for $\BiInt$ given in Sec.~\ref{sec-idea}. We subdivide the construction of $\beta'$ into steps \textbf{(a)}-\textbf{(e)}.

\textbf{(a) Removing redundant cuts.} Call a multicut on $C$ redundant if $C$ appears in its conclusion. Redundant cuts can always be replaced by contractions and weakenings, e.g.
\begin{center}$
\AxiomC{$\Sigma,\{\lbleft{C}\}^s$}
\AxiomC{$\Pi,\lbleft{C},\{\lbright{C}\}^t$}
\RightLabel{$(mcut)$}
\BinaryInfC{$\Sigma,\Pi,\lbleft{C}$}
\DisplayProof
\leadsto
\AxiomC{$\Sigma,\{\lbleft{C}\}^s$}
\RightLabel{$(c)$}
\UnaryInfC{$\Sigma,\lbleft{C}$}
\RightLabel{$(w)$}
\UnaryInfC{$\Sigma,\Pi,\lbleft{C}$}
\DisplayProof$
\end{center}
Note that this replacement does not increase the rank of any multicut below. We therefore assume without loss of generality that all redundant multicuts on $C$ in $\beta$ have been removed.

\textbf{(b) Structuring the proof.} 
Call an inference in $\gamma$ or $\delta$ \emph{critical} if it is a lowermost inference in which an ancestor of multicut is principal. The critial inferences are instances of $(\circ\lright)$ in $\gamma$ and of $(\circ\lleft)$ in $\delta$.  We enumerate the former as $(\circ\lright)_{i\in I}$ and the latter as $(\circ\lleft)_{j\in J}$ and picture them as follows:
\[
\!\!\!\!\!\!\!\!
\usebox\outerleftproof
\usebox\outerrightproof
\]

Here $m,n$ are additional indices for the multiple premises 
of the critical inferences, each of which has a subproof $top_{im}(\gamma)$ or $top_{jn}(\delta)$. $\Lambda_{im}$ and $\Sigma_{jn}$ are the auxiliary formulas. We have also highlighted all further ancestors of multicut ($\{\lbright{C}\}^{p_i}$ and $\{\lbleft{C}\}^{p_j}$) that are not principal in the critical inference. Let us for now assume additionally that
\begin{align*}
\label{eq:nonempty}
    \tag{$\ast$}
    \text{for all }i\in I,\quad \Gamma_i\neq\emptyset
\end{align*}
as this is the more interesting case. We will deal with the remaining case at the very end.




We now identify the \emph{bottom part} $bot(\gamma)$ resp.\ $bot(\delta)$ as the subtree of $\gamma$ ($\delta)$ that contains all sequents that are not above a critical inference. Hence $bot(\gamma)$ is a deduction of $\Gamma,\{\lbright{C}\}^p$ from the conclusions of all critical inferences in $\gamma$, and similarly for $bot(\delta)$. No ancestor of multicut is principal in an inference in $bot(\gamma)$ or $bot(\delta)$. 
We can picture the whole proof as in Figure~\ref{fig:cut} (top diagram).


\begin{figure*}   

\begin{prooftree}
\AxiomC{$\left[\usebox\outerleftproof\right]_i$}
\noLine
\UnaryInfC{\mydots{bot(\gamma)}}
\noLine
\UnaryInfC{$\Gamma,\tms{\lbright{C}}{p}$}
\AxiomC{$\left[\usebox\outerrightproof\right]_j$}
\noLine
\UnaryInfC{\mydots{bot(\delta)}}
\noLine
\UnaryInfC{$\Delta,\tms{\lbleft{C}}{q}$}
\RightLabel{$(mcut)$}
\BinaryInfC{$\Gamma,\Delta$}
\end{prooftree}

\newsavebox\aleftinner
\savebox\aleftinner{
$\AxiomC{\mydots{\text{ axiom expansion}}}
\noLine
\UnaryInfC{$\Gamma_i,(\red{\mathcal{D}})^{p_i},\red{\mathcal{D}}$}
\DisplayProof$
}

\newsavebox\aleftouter
\savebox\aleftouter{
$\AxiomC{$\left[\usebox\aleftinner\right]_i$}
\noLine
\UnaryInfC{\mydots{bot(\gamma)\subst{\red{\mathcal{D}}}{C}}}
\noLine
\UnaryInfC{$\Gamma,(\red{\mathcal{D}})^p$}
\RightLabel{$w,c$}
\UnaryInfC{$\Gamma,\mathcal{D}$}
\DisplayProof
$
}

\newsavebox\arightinner
\savebox\arightinner{
$
\AxiomC{\mydots{top_{ij}}}
\noLine
\UnaryInfC{$\Delta_j,\Gamma_i$}
\RightLabel{$(w)$}
\UnaryInfC{$\Delta_j,(\Gamma_i)^{p_j},\Gamma_i$}
\DisplayProof
$
}

\newsavebox\arightouter
\savebox\arightouter{
$
\AxiomC{$\left[\usebox\arightinner\right]_j$}
\noLine
\UnaryInfC{\mydots{bot(\delta)\subst{\Gamma_i}{C}}}
\noLine
\UnaryInfC{$\Delta,(\Gamma_i)^q$}
\RightLabel{$w,c$}
\UnaryInfC{$\Delta,\mathcal{D}$}
\DisplayProof
$
}
\begin{prooftree}
    \AxiomC{$\left[\usebox\aleftouter\right]_{\mathcal{D}\in\mathbb{D}}$}
    \AxiomC{$\left[\usebox\arightouter\right]_{\mathcal{D}\in\mathbb{D}_i,i\in I}$}
    \dashedLine
    \RightLabel{cuts on $\mathcal{F}$ (Claim~\ref{claim:cutable})}
    \BinaryInfC{$\Gamma,\Delta$}
\end{prooftree}

\caption{(\textit{top diagram}) Proof $\beta$ ending in a multicut with the critical inferences highlighted. (\textit{bottom}) The proof $\beta'$ obtained from $\beta$ via analytic cutting.}
    \label{fig:cut}
\end{figure*}

\textbf{(c1) Substituting in $\gamma$.} Let $\mathcal{E}$ be any set consisting of, for every $i\in I$, some $\ell$-formula $A$ where $\bar{A}\in\Gamma_i$ (note that $\Gamma_i\neq\emptyset$ by (\ref{eq:nonempty})).   Let $bot(\gamma)\subst{\mathcal{E}}{C}$ denote the result of replacing every ancestor of multicut in $bot(\gamma)$ by $\mathcal{E}$:
\[
 \AxiomC{$\left[\Gamma_i,\{\lbright{C}\}^{p_i},\lbright{C}\right]_{i\in I}$}
\noLine
\UnaryInfC{\mydots{bot(\gamma)}}
\noLine
\UnaryInfC{$\Gamma,\{\lbright{C}\}^p$}\
\DisplayProof
\quad\leadsto\quad
 \AxiomC{$\left[\Gamma_i,\mathcal{E}^{p_j},\mathcal{E}\right]_{i\in I}$}
\noLine
\UnaryInfC{\mydots{bot(\gamma)[\mathcal{E}/C]}}
\noLine
\UnaryInfC{$\Delta,\mathcal{E}^p$}
\DisplayProof
\]
Why is the tree on the right again a deduction? For any $i\in I$ we have $\supp{\Gamma_i}\subseteq \context{\circ\lright}$ by inspection of $(\circ\lright)_i$, thus%
\begin{center}
$\mathcal{E}\subseteq\bigcup_{i\in I}\overline{\supp{\Gamma_i}}\subseteq\overline{\context{\circ\lright}}$
\end{center}
Since $\circ$ is inverse rightable we have $\overline{\context{\circ\lright}}\subseteq\context{\RR}$ for any simple rule $\RR$ which has $\lbright{C}\in\context{\RR}$, and so in particular $\mathcal{E}\subseteq\context{\RR}$.  By applying Lemma~\ref{lem:subst} to every rule $\RR$ in $bot(\gamma)$, we ultimately conclude that $bot(\gamma)\subst{\mathcal{E}}{C}$ is a deduction.

\textbf{(c2) Closing the new $\gamma$.} By (\ref{eq:nonempty}), every leaf of $bot(\gamma)\subst{\mathcal{E}}{C}$ contains both an $\ell$-formula $A$ (in $\Gamma_i$) and its inverse $\bar{A}$ (in $\mathcal{E}$), so it can be given a cut-free proof using axiom expansion.

\textbf{(d1) Substituting in $\delta$.} For $i\in I$, $bot(\delta)\subst{\Gamma_i}{C}$ denotes the result of replacing every ancestor of multicut in $bot(\delta)$ by~$\Gamma_i$:
\[
 \AxiomC{$\left[\Delta_j,\{\lbleft{C}\}^{q_j},\lbleft{C}\right]_{j\in J}$}
\noLine
\UnaryInfC{\mydots{bot(\delta)}}
\noLine
\UnaryInfC{$\Delta,\{\lbleft{C}\}^q$}\
\DisplayProof
\leadsto
 \AxiomC{$\left[\Delta_j,(\Gamma_i)^{q_j},\Gamma_i\right]_{j\in J}$}
\noLine
\UnaryInfC{\mydots{bot(\delta)[\Gamma_i/C]}}
\noLine
\UnaryInfC{$\Delta,(\Gamma_i)^q$}\
\DisplayProof
\]
Why is the tree on the right a deduction? We have $\supp{\Gamma_i}\subseteq \context{\circ\lright}$ by inspection of $(\circ\lright)_i$. Since $\circ$ is weakly leftable we furthermore have $\context{\circ\lright}\subseteq\context{\RR}$ for any simple rule $\RR$ which has $\lbleft{C^\circ}\in\context{\RR}$, and so in particular $\supp{\Gamma_i}\subseteq\context{\RR}$. Applying Lemma~\ref{lem:subst} to every rule $\RR$ in $bot(\delta)$, we ultimately conclude that $bot(\delta)\subst{\Gamma_i}{C}$ is a deduction.

\textbf{(d2) Closing the new $\delta$.} Fix a pair $(i,j)\in I\times J$. We obtain a $(d,r)$-reduced proof $top_{ij}$ of $\Gamma_i,\Delta_j$ by a construction that is completely analogous to the two steps of the \princred{} case, only that there are additional indices $i$ and $j$ around. Putting the proofs $top_{ij}$ (plus weakening) on top of $bot(\delta)[\Gamma_i/C]$ we obtain a $(d,r)$-reduced proof.   

\textbf{(e) Putting everything together with cuts.}

The proof $\beta'$ is now constructed as pictured in Figure~\ref{fig:cut}. The top left and top right parts are substituted derivations as described in \textbf{(c)} and \textbf{(d)} (the set $\red{\mathcal{D}}$ will be defined shortly). After that are cuts on formulas from the contexts $\Gamma_i$ ($i\in I$).

Here are the details. Define the set $\mathcal{F}$ as the collection of all formulas appearing in $\Gamma_i$ for some $i\in I$, with their labels stripped off. A \emph{distribution} of $\mathcal{F}$ is any set $\mathcal{D}$ of $\ell$-formulas obtained by labelling the formulas in $\mathcal{F}$. Call $\mathcal{D}$ \emph{$i$-matching} if $\supp{\Gamma_i}\subseteq\mathcal{D}$. Call $\mathcal{D}$ \emph{orthogonal} if it contains a subset $\red{\mathcal{D}}$ that consists of, for every $i$, an $\ell$-formula $A$ where $\bar{A}\in\Gamma_i$. It is easy to see that every distribution is $i$-matching for some $i$, or orthogonal (or both). Now let us denote $\mathbb{D}_i$ ($\mathbb{D}$) the class of $i$-matching (orthogonal) distributions of $\mathcal{F}$. By a simple combinatorial argument, we observe the following: 

\begin{claim}
\label{claim:cutable}
There is a deduction of $\Gamma,\Delta$ from the families of sequents $\Gamma,\mathcal{D}$ ($\mathcal{D}\in\mathbb{D}$) and  $\Delta,\mathcal{D}$ ($\mathcal{D}\in\mathbb{D}_i$, $i\in I$) whose only multicuts are on formulas in $\mathcal{F}$.  
\end{claim}
\noindent This is easy to see bottom-up: If we systematically introduce cuts on all formulas in $\mathcal{F}$ above $\Gamma,\Delta$, we get $2^{|\mathcal{F}|}$ many premises containing every possible distribution of $\mathcal{F}$ (see also Lemma 3.1 in \cite{Tak92}). Now by the previous remark, every distribution is contained in one of the two families.

From this claim it follows that the proof $\beta'$ as in Fig.~\ref{fig:cut} is sound. 
We now show that $\beta'$ is $(d,r)$-reduced.

\begin{claim}
\label{claim:red}
Every formula $A\in\mathcal{F}$ is a subformula of $\Gamma$ or a proper subformula of $C$.
\end{claim}
\begin{proof}
Assume $A$ appears (labelled) in $\Gamma_i$. Follow this occurrence downwards in $\gamma$ via the ancestor relation. This process stops if we reach some $A'$ that is principal in a multicut, and therefore is not the ancestor of any formula in the conclusion. Then, as $\gamma$ is locally analytic, we can choose some $\ell$-formula $A''$ in the conclusion of the multicut that contains $A'$ as a subformula. Now follow $A''$ downwards.
In this way we finally reach some $\ell$-formula $A^*$ in $\Gamma,\{\lbright{C}\}^p$. By construction $A$ is a subformula of $A^*$, so we are done if $A^*\in\Gamma$. Else, $A^*\in\{\lbright{C}\}^p$. As $\lbx{A}$ was not an ancestor of the multicut on $C$---see \textbf{(b)}---by following $A$ downwards we encounter an occurrence $A'$ that was principal in a multicut.
Following down as before we ultimately reach $A^*$, so $A'$ is a subformula of $C$. Since $\gamma$ contains no redundant cuts, $A'\neq C$ and so $A'$ must be a proper subformula of $C$, and hence so is $A$.  \renewcommand\qedsymbol{$\Box$ (Claim)}
\end{proof}
\noindent By Claim~\ref{claim:red}, all the cuts on $\mathcal{F}$ we introduce at the bottom of the proof are either analytic or of degree $<d$.
\begin{claim}
Every non-analytic multicut in $bot(\gamma)\subst{\mathcal{E}_\mathcal{D}}{C}$ and $bot(\delta)\subst{\Gamma_i}{C}$ is on a proper subformula of $C$.
\end{claim}
\begin{proof}
By symmetry, it suffices to consider $bot(\gamma)\subst{\Delta}{C}$. Recall that $bot(\gamma)$ is locally analytic. Hence a multicut in $bot(\gamma)\subst{\Delta}{C}$ can only become non-analytic if it was on a subformula of $C$. Once again $C$ itself is ruled out, as we have removed all redundant multicuts in step $\textbf{(a)}$.
\renewcommand\qedsymbol{$\Box$ (Claim)}
\end{proof}
\noindent It follows that $\beta'$ is $(d,r)$-reduced.

Let us finally demonstrate the much simpler construction in the case that (\ref{eq:nonempty}) fails, meaning that $\Gamma_{i_0}=\emptyset$ for some $i_0\in I$. In this case the transformed proof is essentially just the right part of Fig.~\ref{fig:cut} (bottom diagram) where $i=i_0$, and no analytic cuts need to be introduced. Indeed:
We construct the proofs $top_{i_0j}$ ($j\in J$) and $bot(\delta)[\Gamma_{i_0}/C]=bot(\delta)[\emptyset/C]$ as before. But now the endsequent of $bot(\delta)[\emptyset/C]$ is $\Delta$, so we can go directly to $\Gamma,\Delta$ via weakening.

\trans{\rightanalytic}{$C=C^\circ$, $\circ$ is weakly rightable and inverse leftable}
Symmetric to \leftanalytic.
\subsection*{(B) Completeness of the reduction steps}
We show that always at least one reduction step applies.
If $C$ is a variable, then \renaming{} applies. 
Assume $C=C^\circ$ for some class 2 connective $\circ$.
If $\circ$ is left and right invertible, \inversion{} applies. If it is weakly leftable and inverse rightable, then \leftanalytic{} applies. 
If it is weakly rightable and inverse leftable, then \rightanalytic{} applies.
%
Two cases are left.
Assume that $\circ$ is leftable and weakly rightable. We make a case distinction: If no ancestor of the multicut is principal in $\RR(\delta)$, then \rightshift{} is applicable. If not, then $\RR(\delta)$ is $(\circ\lleft)$ and therefore $\supp{\Delta}\subseteq\context{\circ\lleft}$. If moreover some ancestor of multicut is principal in $\RR(\gamma)$, then \princred{} applies. If on the other hand no ancestor of multicut is principal in $\RR(\gamma)$, then \leftshift{} is applicable.
The case  $\circ$ is rightable and weakly leftable is symmetric.

We have shown that there is always some reduction that applies. This concludes the proof of the Reduction Lemma.

\begin{remark}
Each step in cut-restriction only ``locally'' improves the analyticity of the proof. That is, if an uppermost non-analytic cut has been made analytic by a reduction step it might still not be \emph{globally analytic}, i.e.\ its cut formula might still not be a subformula of the endsequent. This is because of other non-analytic cuts below it. But the algorithm will revisit such cuts (once they become non-analytic through substitution), and only when all cuts are (locally) analytic, they also become 
globally analytic. 
\end{remark}

\section{Applications of the Main Theorem}
\label{sec:applications}
We illustrate the practicality of our method by showing that 
various calculi have the analytic cut property. This amounts to demonstrating that their connectives are in class 2, as the other requirements---axiom expansion, principal case reductions and consistency---are readily verified by standard methods. 


\begin{corollary}
    $\SF$ has the analytic cut property.
\end{corollary}
\begin{proof}
The only rule of $\SF$ with a context restriction is $(\Box\lright)$, and it admits only boxed formulas as contexts. Hence all connectives $\circ\neq\Box$ are left- and right-shiftable, and therefore class~1. Furthermore $\Box$ is inverse rightable and weakly leftable (Ex.~\ref{ex:modalities}), and so the whole calculus is class 2.
\end{proof}

\emph{Multi-modal $\SF$} is the extension of $\LK$ with multiple modalities $\Box_1,\ldots,\Box_n$. The rules for $\Box_i$ are obtained by indexing with $i$ the modalities in the $\SF$ rules
$(\Box\lleft)$, and $(\Box\lright)$, 
so that $\context{\Box_i\lright}$ contains  all formulas prefixed by $\Box_i$. 

By the same argument as above, we obtain:

\begin{corollary}
    Multi-modal $\SF$ has the analytic cut property.
\end{corollary}

\begin{corollary}
    $\BiInt$ has the analytic cut property.
\end{corollary}
\begin{proof}
 In $\BiInt$, $\bot$, $\land$ and $\lor$ are left- and right-invertible by Lemma~\ref{lem:invertible},
 and thus class 1. The implication connectives are class 2: $\to$ is weakly leftable and inverse rightable, whereas $\coimp$ is weakly rightable and inverse leftable (Ex.~\ref{BI-Int}). 
\end{proof}

\begin{corollary}
$\GF$ (Ex.~\ref{ex:modalities}) has the analytic cut property.
\end{corollary}
\begin{proof}
In $\GF$ the connectives $\bot$, $\land$ and $\lor$ are invertible. $\to$ is rightable and weakly leftable and therefore class 1 as well. Finally $\Box$ is weaky leftable and inverse rightable (Ex.~\ref{ex:modalities}), hence the calculus is class 2.
\end{proof}

\begin{corollary}
    $\BiInt^{\SF}$ (Ex.~\ref{ex:modalities}) has the analytic cut property.
\end{corollary}
\begin{proof} Once more $\bot$, $\land$ and $\lor$ are invertible. $\to$ is weakly leftable and inverse rightable, whereas $\coimp$ is weakly rightable and inverse leftable. Finally $\Box$ is weakly leftable and inverse rightable (Ex.~\ref{ex:modalities}). So the calculus is class 2.
\end{proof}

\section{Cut-elimination}
\label{sec:classone}
We have shown that every class 2 standard calculus has the analytic cut property. We show here that under further assumptions we can recover cut-elimination. 

In particular, we are now required to eliminate cuts on variables. This was not needed for cut-restriction, as
%
such cuts can be made analytic by a simple substitution. For elimination, cuts on variables must be treated in essentially the same way as cuts on compound formulas, and we therefore introduce an analogous property to leftability/rightability.

\begin{definition}
$\SC$ satisfies \emph{leftable variables} (\emph{rightable variables}) if for every simple rule $\RR$, if $\lbleft{x}\in\context{\RR}$ ($\lbright{x}\in\context{\RR}$) for some variable $x$, then $\context{\RR}$ has no context restriction.
\end{definition}

\begin{example}
Maehara's calculus satisfies rightable variables: The only rule with a context restriction is $(\to\lright)$, and we have $\lbright{x}\notin\context{\to\lright}$. It does not satisfy leftable variables.
\end{example}

\begin{definition}
    A standard calculus is \emph{class 1} if it satisfies principal case reductions, axiom expansion, leftable or rightable variables, and every connective in it is class 1.
\end{definition}

\begin{theorem}\label{thm-classone}
    Every class 1 calculus admits 
    cut-elimination.
\end{theorem}
\begin{proof}[Proof (sketch)]
    Similar to Th.~\ref{thm:classtwo}. Since all connectives are class~1, the reductions \leftanalytic{} and \rightanalytic{} which would introduce analytic cuts
    never need to be applied. Also avoid the reduction \renaming. Instead, shift cuts on variables upwards by adopting \leftshift{} (if $\SC$ has rightable variables) 
    or \rightshift{} 
    (if $\SC$ has leftable variables) 
    until they are principal in one premise of the cut. In this case the premise is an initial sequent, and the cut can be omitted.
\end{proof}

\begin{remark}
\label{rem:class12}
Checking that a calculus is class 1 or 2 is not a modular task: If we extend, say, a class 1 calculus with a new connective we will have to ``re-evaluate'' the status of all old connectives. This failure of modularity is to be expected: For example, both fragments of $\BiInt$ in the languages $\{\land,\lor,\bot,\to\}$ and $\{\land,\lor,\bot,\coimp\}$ are class 1 and thus satisfy 
cut-elimination, but $\BiInt$ itself is only class 2.
\end{remark}

\section{Conclusions}

We introduced cut-restriction, an algorithm transforming proofs with arbitrary cuts into proofs with analytic cuts. The result is obtained through language-independent sufficient conditions. Our methodology encompasses existing results in a uniform way, and yields novel results about the analytic cut property.
Moreover, we have identified the strengthening of the sufficient conditions that implies 
cut-elimination, thus showing that the latter is as a special case of cut-restriction.

\textit{Future work.} Maehara's method for the Craig interpolation property is not hindered by analytic cuts (see, e.g.,~\cite{KowOno17,OnoSano}). This motivates a general investigation into Craig interpolation for calculi that have the analytic cut property.

From the computational interpretation point of view, it would be interesting to investigate the meaning of our procedure within the Curry-Howard paradigm, and its possible connections with the notion of partial evaluation. 

Another research direction would be to generalize our conditions. Notice indeed that the modal calculus {\bf K} is not a standard sequent calculus under the definition presented here, as its modal rule has arbitrarily many principal formulas. The restriction to logical rules having a single principal formula (``simple rules'') served to simplify the notation in the main proof, but the argument can be extended to the case of {\bf K}. 
There are other rules whose form is not analytic that we would like to encompass, e.g. the peculiar rule of the modal logic {\bf B}, and bi-intuitionistic stable tense logic {\bf BiSKt} (known via semantic methods to have the analytic cut property~\cite{OnoSano}).
A further investigation would be to consider substructural logics where weakening and contraction might not be present. 
We are not aware of a substructural logic without cut-elimination that is complete for analytic cuts, but~\cite{CiaLanRam21} presents many substructural logics with a modified  subformula property.

Takano obtains a relaxation of the subformula property for several modal logics via semantics:  $K5$, $K5D$,  $S4.2$, $KD\#$~\cite{Tak01,Tak19,Tak20}. We would like to extend cut-restriction to such modified subformula properties. A broader aim is a \textit{general classification of logics in terms of
their modified subformula properties under cut-restriction}. This is a reimagining of structural proof theory: instead of constructing generalisations of the sequent calculus to get cut-elimination, aim for modifications of the subformula property. 
%
This would mean a single target for theory and applications 
including theorem proving, proof-assistants, and meta-theoretic argumentation.

\section*{Acknowledgement}

Work partially supported by the FWF project P33548 and EPSRC projects EP/S013008/1,  EP/R006865/1.

\bibliographystyle{abbrv}
\bibliography{mybib}

\begin{thebibliography}{10}

\bibitem{Avr96}
A.~Avron.
\newblock The method of hypersequents in the proof theory of propositional
  non-classical logics.
\newblock In {\em Logic: from foundations to applications ({S}taffordshire,
  1993)}, pages 1--32. Oxford Univ. Press, New York, 1996.

\bibitem{AvrLahav}
A.~Avron and O.~Lahav.
\newblock A unified semantic framework for fully structural propositional
  sequent systems.
\newblock {\em ACM Trans. Comput. Logic}, 14(4), 2013.

\bibitem{BaazHLRS08}
M.~Baaz, S.~Hetzl, A.~Leitsch, C.~Richter, and H.~Spohr.
\newblock {CERES:} an analysis of {F}{\"{u}}rstenberg's proof of the infinity
  of primes.
\newblock {\em Theor. Comput. Sci.}, 403(2-3):160--175, 2008.

\bibitem{BaazL06}
M.~Baaz and A.~Leitsch.
\newblock Towards a clausal analysis of cut-elimination.
\newblock {\em J. Symb. Comput.}, 41(3-4):381--410, 2006.

\bibitem{Bel82}
N.~D. Belnap, Jr.
\newblock Display logic.
\newblock {\em J. Philos. Logic}, 11(4):375--417, 1982.

\bibitem{CiaLanRam21}
A.~Ciabattoni, T.~Lang, and R.~Ramanayake.
\newblock Bounded-analytic sequent calculi and embeddings for hypersequent
  logics.
\newblock {\em J. Symb. Log.}, 86(2):635--668, 2021.

\bibitem{DAgoMon94}
M.~D'Agostino and M.~Mondadori.
\newblock The taming of the cut. {C}lassical refutations with analytic cut.
\newblock {\em J. of Logic and Computation}, 4:285--319, 1994.

\bibitem{Fit78}
M.~Fitting.
\newblock Subformula results in some propositional modal logics.
\newblock {\em Studia Logica}, 37(4):387--391, 1978.

\bibitem{KowOno17}
T.~Kowalski and H.~Ono.
\newblock Analytic cut and interpolation for bi-intuitionistic logic.
\newblock {\em The Review of Symbolic Logic}, 10(2):259--283, 2017.

\bibitem{Min68}
G.~E. Mints.
\newblock Some calculi of modal logic.
\newblock {\em Trudy Mat. Inst. Steklov}, 98:88--111, 1968.

\bibitem{Ohn59}
M.~Ohnishi and K.~Matsumoto.
\newblock Gentzen method in modal calculi. ii.
\newblock {\em Osaka Mathematical Journal}, 11(2):115--120, 1959.

\bibitem{Ono77}
H.~Ono.
\newblock On some intuitionistic modal logics.
\newblock {\em Publications of the Research Institute for Mathematical
  Sciences}, 13:687--722, 1977.

\bibitem{OnoSano}
H.~Ono and K.~Sano.
\newblock Analytic cut and {M}ints’ symmetric interpolation method for
  {Bi}-intuitionistic tense logic.
\newblock In {\em Advances in Modal Logic}, pages 601--624. College
  Publications, 2022.

\bibitem{Pfe10}
F.~Pfenning.
\newblock Lecture notes on sequent calculus.
\newblock {\em Lecture Notes for the Carnegie Mellon University course}, pages
  15--816, 2010.

\bibitem{PintoU18}
L.~Pinto and T.~Uustalu.
\newblock A proof-theoretic study of bi-intuitionistic propositional sequent
  calculus.
\newblock {\em J. Log. Comput.}, 28(1):165--202, 2018.

\bibitem{Rau74}
C.~Rauszer.
\newblock A formalization of the propositional calculus of {$H-B$} logic.
\newblock {\em Studia Logica}, 33:23--34, 1974.

\bibitem{Restall}
G.~Restall.
\newblock {\em An Introduction to Substructural Logics}.
\newblock Routledge, 2000.

\bibitem{Smu68}
R.~M. Smullyan.
\newblock Analytic cut.
\newblock {\em J. Symbolic Logic}, 33:560--564, 1968.

\bibitem{CH06}
M.~Sorensen and P.~Urzyczyn.
\newblock {\em Lectures on the Curry-Howard isomorphism}.
\newblock Elsevier, Amsterdam, 2006.

\bibitem{Tak92}
M.~Takano.
\newblock Subformula property as a substitute for cut-elimination in modal
  propositional logics.
\newblock {\em Mathematica japonica}, 37:1129--1145, 1992.

\bibitem{Tak01}
M.~Takano.
\newblock A modified subformula property for the modal logics {K5} and {K5D}.
\newblock {\em Bulletin of the Section of Logic}, 30, 01 2001.

\bibitem{Tak19}
M.~Takano.
\newblock A modified subformula property for the modal logic {S}4.2.
\newblock {\em Bulletin of the Section of Logic}, 48:19--28, 2019.

\bibitem{Tak20}
M.~Takano.
\newblock New modification of the subformula property for a modal logic.
\newblock {\em Bulletin of the Section of Logic}, 49, 08 2020.

\bibitem{Takeuti:87}
G.~Takeuti.
\newblock {\em Proof theory}, volume~81 of {\em Studies in Logic and the
  Foundations of Mathematics}.
\newblock North Holland, Amsterdam, 1987.

\bibitem{Val83}
S.~Valentini.
\newblock The modal logic of provability: cut-elimination.
\newblock {\em J. Philos. Logic}, 12(4):471--476, 1983.

\end{thebibliography}
\section*{Appendix: proof of Lemma~\ref{lem:invertible}}

We prove that under the premises of Lemma~\ref{lem:invertible}, 
$\circ$ is \emph{size-preserving left-invertible}, 
that strenghtens 
Def.~\ref{def:invertibility} by the clause: 
the number of nodes in $\beta'$ is no larger than the number of nodes in $\beta$. 
The proof is by induction on the number of nodes in $\beta$. Let $\RR$ be the last inference in $\beta$.

\begin{enumerate}
\item \textit{$\RR$ is a multicut on $C$.}\\ Without loss of generality, the end of $\beta$ runs as follows where $\Gamma=\Gamma_1,\Gamma_2$ (the case that $\lbleft{C^\circ}$ is in the conclusion of $\delta_1$ is symmetric):
\[
\infer[(mcut)]{\Gamma_1,\Gamma_2,\lbleft{C^\circ}}
    	{
    	\deduce{\Gamma_1,\{\lbleft{C^\circ}\}^{p}}{\mydots{\delta_1}}
    	&
    	\deduce{\Gamma_2,\tms{\lbleft{C^\circ}}{q},\lbleft{C^\circ}}{\mydots{\delta_2}}
    	}
\]
Applying the induction hypothesis $(q+1)$-many times to $\delta_2$ we obtain a proof $\delta_2'$ of $\Gamma_2,(\Lambda_m)^{q+1}$. Note that the induction hypothesis can be applied multiple times because the size of the obtained proofs is always bounded by number of nodes in $\delta_1$. Now using contraction and weakening we obtain the proof $\beta'$ of $\Gamma_1,\Gamma_2,\Lambda_m$.

\item \textit{$\RR$ is not a multicut on $C$ and if $\RR$ is a simple rule, then $\lbleft{C^\circ}$ is not principal in it.}\\
The end of $\beta$ runs as follows
\[
    \infer{\Gamma,\lbleft{C^\circ}}
    	{
    	\deduce{\Gamma_1,\tms{C^\circ}{q_1}}{\mydots{\delta_1}}
    	&
    	\ldots
    	&
    	\deduce{\Gamma_n,\tms{C^\circ}{q_n}}{\mydots{\delta_n}}
    	}
\]
with all immediate ancestors of $\lbleft{C^\circ}$ highlighted in the premises. We construct $\beta'$ as follows
\[
    \infer{\Gamma,\Lambda_m}
    	{
    	\deduce{\Gamma_1,\tms{\Lambda_m}{q_1}}{\mydots{\delta_1'}}
    	&
    	\ldots
    	&
    	\deduce{\Gamma_n,\tms{\Lambda_m}{q_n}}{\mydots{\delta_n'}}
    	}
\]
where $\delta_i'$ is obtained from $\delta_i$ by $q_i$ applications of the induction hypothesis. By assumption and Lemma~\ref{lem:subst}, the lowest inference in $\beta'$ is again $\RR$.

If $\RR$ is a multicut then it can cease to be analytic moving from $\beta$ to $\beta'$, but only if the cut formula was a subformula of $C^\circ$. As we have ruled out the case that the cut formula is $C^\circ$, it must be a proper subformula.

\item \textit{$\lbleft{C^\circ}$ is principal.} \\
The end of $\beta$ runs as follows:
\begin{prooftree}
    \AxiomC{\mydots{\delta_1}}
    \noLine
    \UnaryInfC{$\Gamma,\Lambda_1'$}
    \AxiomC{$\ldots$}
    \AxiomC{\mydots{\delta_M}}
    \noLine
    \UnaryInfC{$\Gamma,\Lambda_M'$}
    \RightLabel{$(\circ\lleft)$}
    \TrinaryInfC{$\Gamma,\lbleft{C^\circ}$}
\end{prooftree}
By the uniqueness assumption $\Lambda_m'=\Lambda_m$. But this means we can simply take $\delta_m$ as $\beta'$.
\end{enumerate}

\end{document}